\font \sevenrm=cmr7
\font \fiverm=cmr5
\documentclass[12pt, english]{amsart}
\input epsf.sty
\usepackage{amsfonts}
\usepackage{amssymb}
\usepackage{amsmath}
\usepackage{amsthm}
\usepackage{amscd}
\usepackage[all]{xy} 
\usepackage{epsfig,exscale}
\usepackage{color}
\usepackage{graphicx}
\usepackage{xspace}
\usepackage{axodraw4j}
\usepackage{colordvi}

\newcommand{\nc}{\newcommand}

{\everymath{\displaystyle\everymath{}}\array}%
{\endarray}
\newtheorem{theorem}{Theorem}
\newtheorem{definition}{Definition}
\newtheorem{corollary}{Corollary}
\newtheorem{lemma}{Lemma}
\newtheorem{proposition}{Proposition}
\newtheorem{ex}{Example}
\newtheorem{remark}{Remark}
\nc{\comment}[1]{[[{\tt #1}]] }
\nc{\Cal}[1]{{\mathcal {#1}}}
\nc{\mop}[1]{\mathop{\hbox {\rm #1} }\nolimits}
\nc{\gmop}[1]{\mathop{\hbox {\bf #1} }\nolimits}

\nc{\smop}[1]{\mathop{\hbox {\sevenrm #1} }\nolimits}
\nc{\ssmop}[1]{\mathop{\hbox {\fiverm #1} }\nolimits}
\nc{\mopl}[1]{\mathop{\hbox {\rm #1} }\limits}
\def\dbar{d\hskip-3pt \raise 4pt\hbox{-}}
\nc{\smopl}[1]{\mathop{\hbox {\sevenrm #1} }\limits}
\nc{\ssmopl}[1]{\mathop{\hbox {\fiverm #1} }\limits}
\nc{\frakg}{{\frak g}}
\nc{\g}[1]{{\frak {#1}}}
\def \restr#1{\mathstrut_{\textstyle |}\raise-6pt\hbox{$\scriptstyle #1$}}
\def \srestr#1{\mathstrut_{\scriptstyle |}\hbox to
-1.5pt{}\raise-4pt\hbox{$\scriptscriptstyle #1$}}
\nc{\wt}{\widetilde} \nc{\wh}{\widehat}
\nc{\redtext}[1]{\textcolor{red}{#1}}
\nc{\bluetext}[1]{\textcolor{blue}{#1}}
\nc\fleche[1]{\mathop{\hbox to #1 mm{\rightarrowfill}}\limits}
\nc{\ignore}[1]{}
\def\semi{\mathrel{\times}\kern -.85pt\joinrel\mathrel{\raise
1.4pt\hbox{${\scriptscriptstyle |}$}}}
\nc\R{{\mathbb R}}
\nc\N{{\mathbb N}}
\nc\inver{^{-1}}
\nc\point{\hbox{\bf .}}
\nc\un{\hbox{\bf 1}}
\renewcommand{\Re}{\operatorname{Re}}

\newcommand{\tr}{\textrm{tr}}

\def\graphearetenor{\,{\scalebox{0.25}{
\begin{picture}(130,6) (175,-221)
\SetWidth{3.0}
\SetColor{Black}
\Line(176,-218)(304,-218)
\end{picture}
}}\,}

\def\graphearetecurcroixun{\,{\scalebox{0.25}{
\begin{picture}(130,90) (111,-177)
\SetWidth{3.0}
\SetColor{Black}
\Photon(112,-164)(240,-164){7.5}{6}
\Line(191.998,-179.998)(160.002,-148.002)\Line(160.002,-179.998)(191.998,-148.002)
\Text(176,-212)[lb]{\Huge{\Black{$1$}}}
\end{picture}
}}\,}

\def\graphearetenorcroixzero{\,{\scalebox{0.25}{
\begin{picture}(130,34) (95,-143)
\SetWidth{3.0}
\SetColor{Black}
\Line(96,-126)(224,-126)
\Line(144.002,-110.002)(175.998,-141.998)\Line(175.998,-110.002)(144.002,-141.998)
\Text(156,-172)[lb]{\Huge{\Black{$0$}}}
\end{picture}
}}\,}

\def\graphearetenorcroixun{\,{\scalebox{0.25}{
\begin{picture}(130,34) (95,-143)
\SetWidth{3.0}
\SetColor{Black}
\Line(96,-126)(224,-126)
\Line(144.002,-110.002)(175.998,-141.998)\Line(175.998,-110.002)(144.002,-141.998)
\Text(156,-172)[lb]{\Huge{\Black{$1$}}}
\end{picture}
}}\,}

\def\graphearetecur{\,{\scalebox{0.25}{
\begin{picture}(133,20) (104,-209)
\SetWidth{3.0}
\SetColor{Black}
\Photon(105,-201)(236,-197){7.5}{7}
\end{picture}
}}\,}

\def\graphev{\,{\scalebox{0.25}{
\begin{picture}(66,34) (191,-191)
    \SetWidth{2.0}
    \SetColor{Black}
    \Line(192,-174)(224,-174)
    \Line(224,-174)(256,-158)
    \Line(224,-174)(256,-190)
\end{picture}}}\,}

\def\graphevcur{\,{\scalebox{0.25}{
\begin{picture}(160,166) (97,-122)
\SetWidth{2.0}
\SetColor{Black}
\Photon(98,-99)(195,-97){7.5}{5}
\Line(195,-99)(256,-37)
\Line(196,-100)(255,-161)
\end{picture}
}}\,}

\def\grapheexmdel{\,{\scalebox{0.25}{
\begin{picture}(162,52) (175,-209)
\SetWidth{3.0}
\SetColor{Black}
\Arc(256,-222)(16,270,630)
\Arc(256,-192)(34,-61.928,241.928)
\Line(288,-190)(336,-190)
\Line(224,-190)(176,-190)
\end{picture}
}}\,}

\def\graphecerc{\,{\scalebox{0.25}{
\begin{picture}(130,34) (191,-287)
\SetWidth{3.0}
\SetColor{Black}
\Arc(256,-270)(16,270,630)
\Line(272,-270)(320,-270)
\Line(240,-270)(192,-270)
\end{picture}
}}\,}

\def\graphecerccroizero{\,{\scalebox{0.25}{
\begin{picture}(162,63) (175,-101)
\SetWidth{3.0}
\SetColor{Black}
\Arc(256,-75)(35.777,243,603)
\Line(288,-75)(336,-75)
\Line(224,-75)(176,-75)
\Line(240.002,-122.998)(271.998,-91.002)\Line(240.002,-91.002)(271.998,-122.998)
\Text(256,-139)[lb]{\Huge{\Black{$0$}}}
\end{picture}
}}\,}

\def\graphecerccroiun{\,{\scalebox{0.25}{
\begin{picture}(162,63) (175,-101)
\SetWidth{3.0}
\SetColor{Black}
\Arc(256,-75)(35.777,243,603)
\Line(288,-75)(336,-75)
\Line(224,-75)(176,-75)
\Line(240.002,-122.998)(271.998,-91.002)\Line(240.002,-91.002)(271.998,-122.998)
\Text(256,-139)[lb]{\Huge{\Black{$1$}}}
\end{picture}
}}\,}

\def\graphedeltasansindice{\,{\scalebox{0.20}{
\begin{picture}(226,68) (111,-161)
    \SetWidth{2.0}
    \SetColor{Black}
    \Arc(224,-142)(48,180,540)
    \Line(192,-110)(192,-174)
    \Line(240,-94)(240,-190)
    \Line(272,-142)(336,-142)
    \Line(176,-142)(112,-142)
\end{picture}
}}\,}

\def\graphedeltacontratddd{\,{\scalebox{0.20}{
\begin{picture}(96,57) (336,-297)
    \SetWidth{2.0}
    \SetColor{Black}
\Arc(382.723,-358.862)(114.921,68.415,112.902)
    \Arc(379,-191.833)(131.078,-108.688,-67.576)
    \Line(354,-248)(354,-320)
    \Line(412,-249)(415,-317)
\end{picture}
  }}\,}
  
\def\graphedeltacontratdddd{\,{\scalebox{0.20}{
\begin{picture}(168,29) (184,-180)
    \SetWidth{2.0}
    \SetColor{Black}
    \Arc(243,-176)(22.627,135,495)
    \Arc(288,-175)(22.627,135,495)
    \Line(312,-176)(351,-176)
    \Line(221,-177)(185,-177)
\end{picture}
  }}\,}

 \def\qedj{\,{\scalebox{0.20}{ \begin{picture}(181,66) (218,-178)
    \SetWidth{3.0}
    \SetColor{Black}
    \PhotonArc(307.786,-179.786)(22.955,14.598,173.03){7.5}{3.5}
    \PhotonArc[clock](306,-170.142)(49.152,-175.498,-364.502){7.5}{8.5}
    \Line(219,-175)(398,-175)
  \end{picture}}}\,}
  
\def\qeda{\,{\scalebox{0.20}{
\begin{picture}(98,35) (271,-286)
    \SetWidth{3.0}
    \SetColor{Black}
    \PhotonArc(320,-277)(17.889,-26.565,206.565){7.5}{4.5}
    \Line(272,-285)(368,-285)
\end{picture}}}\,}

\def\qedcoiz{\,{\scalebox{0.20}
  { \begin{picture}(258,93) (191,-180)
    \SetWidth{3.0}
    \SetColor{Black}
    \PhotonArc(328,-187)(91.214,15.255,164.745){7.5}{12.5}
    \Line(192,-163)(448,-163)
    \Line(335.998,-178.998)(304.002,-147.002)\Line(304.002,-178.998)(335.998,-147.002)
    \Text(315,-210)[lb]{\Huge{\Black{$0$}}}
  \end{picture}}}\,}
  
\def\qedcoiu{\,{\scalebox{0.20}
  { \begin{picture}(258,93) (191,-180)
    \SetWidth{3.0}
    \SetColor{Black}
    \PhotonArc(328,-187)(91.214,15.255,164.745){7.5}{12.5}
    \Line(192,-163)(448,-163)
    \Line(335.998,-178.998)(304.002,-147.002)\Line(304.002,-178.998)(335.998,-147.002)
    \Text(315,-210)[lb]{\Huge{\Black{$1$}}}
  \end{picture}}}\,}

\def\grphopropg{\,{\scalebox{0.3}{   
    \begin{picture}(100,144) (223,-80)
    \SetWidth{2.5}
    \SetColor{Black}
    \Arc[arrow,arrowpos=0.5,arrowlength=5,arrowwidth=5,arrowinset=0.2,clock](344,-66.667)(56.063,177.274,2.726)
    \Arc[arrow,arrowpos=0.5,arrowlength=5,arrowwidth=5,arrowinset=0.2,clock](344,-56)(56.569,-8.13,-171.87)
    \Line[arrow,arrowpos=0.5,arrowlength=5,arrowwidth=5,arrowinset=0.2](224,-64)(288,-64)
    \Line[arrow,arrowpos=0.5,arrowlength=5,arrowwidth=5,arrowinset=0.2](400,-64)(480,-64)
    \Text(256,-48)[lb]{\huge{\Black{$p$}}}
    \Text(345,0)[lb]{\huge{\Black{$k$}}}
    \Text(440,-48)[lb]{\huge{\Black{$p$}}}
    \Text(336,-142)[lb]{\huge{\Black{$k - p$}}}
  \end{picture}}}\,}

\def\diagramme #1{\vskip 4mm \centerline {#1} \vskip 4mm}
\begin{document}
\title{
{Doubling bialgebras of graphs and Feynman rules}}
\author{Mohamed Belhaj Mohamed}
\address{{Universit\'e Blaise Pascal,
         laboratoire de math\'ematiques UMR 6620,
         63177 Aubi\`ere, France.}\\
         {Laboratoire de math\'ematiques physique fonctions sp\'eciales et applications, universit\'e de sousse, rue Lamine Abassi 4011 H. Sousse,  Tunisie}}     
         \email{Mohamed.Belhaj@math.univ-bpclermont.fr}
       
\date{October, 15 th 2014}
\noindent{\footnotesize{${}\phantom{a}$ }}
\begin{abstract}
In this article, we define a doubling procedure for the bialgebra of specified Feynman graphs introduced in a previous paper \cite {DMB}. This is the vector space generated by the pairs $(\bar \Gamma, \bar \gamma)$ where $\bar \Gamma$ is a locally $1PI$ specified graph of a perturbation theory $\Cal T$ with $\bar \gamma \subset \bar \Gamma$ locally $1PI$ and where $\bar \Gamma / \bar \gamma $ is a specified graph of $\Cal T$. We also define a convolution product on the characters of this new bialgebra with values in an endomorphism algebra, equipped with a commutative product compatible with the composition. We then express in this framework the renormalization as formulated by A. Smirnov \cite [\S 8.5, 8.6] {Sm}, adapting the approach of A. Connes and D. Kreimer for two renormalization schemes: the minimal renormalization scheme and the Taylor expansion scheme. Finally, we determine the finite parts of Feynman integrals using the BPHZ algorithm after dimensional regularization procedure, by following the approach by P. Etingof \cite{PE} (see also \cite{RM}).
\end{abstract}
\maketitle

\tableofcontents
\section{Introduction}
In this note, we are interested in Feynman rules, given by integration of some type of functions with respect to the internal momenta. We recall the construction of the bialgebra $\wt{\Cal H}_{\Cal T}$ and the Hopf algebra ${\Cal H}_{\Cal T}$ of specified Feynman graphs associated to some perturbative theory $\Cal T$ \cite {DMB}. We define the doubling $\wt{\Cal D}_{\Cal T}$ of this bialgebra. This is the vector space generated by the pairs $(\bar \Gamma, \bar \gamma)$ where $\bar \Gamma$ is a locally $1PI$ specified graph of the theory $\Cal T$ with $\bar \gamma \subset \bar \Gamma$ locally $1PI$ and where $\bar \Gamma / \bar \gamma $ is the corresponding specified graph and we consider the following coproduct:
$$
\Delta (\bar\Gamma, \bar\gamma ) = \sum_{\substack{\bar\delta \subseteq \bar\gamma \\ \bar\gamma / \bar\delta \in \Cal T }} ( \bar \Gamma, \bar\delta )  \otimes ( \bar\Gamma / \bar\delta, \bar\gamma / \bar\delta).
$$
We define then the convolution product $\divideontimes$ on the group $G$ of characters of the Hopf algebra ${\Cal D}_{\Cal T}$ with values in an endomorphism algebra $\mop {End} \wt {\Cal B}$ (where $\wt {\Cal B}$ is defined in \S 9.1), equipped with a commutative product $\bullet$ compatible with the composition $\circ$, which takes into account the dependence of the external momenta:
$$
 \varphi \divideontimes \psi :=  \diamond  ( \varphi \otimes \psi ) \Delta,
$$
where $\diamond$ is the opposite of the composition product $\circ$. In other words, for all specified graphs $\bar\gamma$, $\bar\Gamma$ such that $\bar\gamma \subset \bar\Gamma$ we have:
$$
(\varphi \divideontimes \psi) (\bar\Gamma, \bar\gamma ) = \sum_{\substack{\bar\delta \subseteq \bar\gamma \\ \bar\gamma / \bar\delta \in \Cal T }}\psi( \bar\Gamma / \bar\delta, \bar\gamma / \bar\delta)\circ \varphi( \bar \Gamma, \bar\delta ).
$$
We then retrieve the renormalization procedure as formulated by A. Smirnov \cite [\S 8.5, 8.6] {Sm}, adopting the A. Connes and D. Kreimer approach for two renormalization schemes: the minimal renormalization scheme and the Taylor expansion scheme.\\
We then get interested on Feynman integrals which are generally divergent. We use an algebraic approach to reinterpret Smirnov's approach \cite [\S 8] {Sm} in the Connes-Kreimer formalism, using the Hopf algebra ${\Cal D}_{\Cal T}$ and a target algebra which is no longer commutative, because of composition of operators.
We use the dimensional regularization procedure for constructing the $D$-dimensional integrals, which consists in writing the divergent integrals that we have to regularize in such a way that the dimension of the physical space-time $d$ can be replaced by any complex number $D$. We follow the approach of P. Etingof \cite{PE} (see also \cite{RM}). Let $V$ be the $d$-dimensional space-time and $\Gamma$ a Feynman graph with $m$ external edges and corresponding momenta $q_1, \ldots, q_m \in V$, and with $n - m$ loops $(n \geq m)$. The amplitude of this graph can be written as: 
$$
I_{( q_1, \ldots, q_m )} (f)  :=  \int_{ V^{n-m}} f( q_1, \ldots, q_n ) dq_{m+1} \ldots dq_n.
$$

\noindent Let $\Gamma$, $\gamma$ and $\delta$ be three Feynman graphs such that  $\delta \subseteq \gamma \subseteq \Gamma$. We denote by:
$$ E := \Cal E (\Gamma),    \;\;\;\;\;\;\;\; F := \Cal E (\Gamma/\delta) \;\;\;\;\;\text{and} \;\;\;\;\; G:= \Cal E (\Gamma/\gamma) = \Cal E (\Gamma/\delta \Big/ \gamma/\delta),$$
the vector spaces spanned by the half-edges of the three graphs $\Gamma$, $\Gamma/\delta$ and $\Gamma/\gamma$ respectively. We denote by $S^2 E^*$ the vector space of symmetric bilinear forms on $E$, by $\bar S_+^2 E^*$ the subset of positive semi-definite bilinear forms on $E$, by $S_+^2 E^*$ the subset of positive definite bilinear forms on $E$, and by $\Cal S (\bar S_+^2 E^*)$ and $\Cal S (S^2 E^*)$ the two spaces of Schwartz functions on $\bar S_+^2 E^*$ and $ S^2 E^*$ respectively. We adopt a similar notations for $F$ and $G$. We view an n-tuple $q:= (q_1, \ldots, q_n)$, with $q_j \in \mathbb R^d$, as an element of the vector space $\mop{Hom}(E, \mathbb R^d)$. The m-tuple $(q_1, \ldots, q_m)$ is nothing but the restriction $ q_{|F} \in \mop{Hom}(F, \mathbb R^d)$. The bilinear form $q^* (\beta) \in \bar S_+^2 E^*$  is obtained by pulling back $\beta \in (S_+^2 \mathbb R^d)^*$ along $q$.
Then we define the integral of a function $f$ that is defined on $\bar S_+^2 E^*$ by:
$$
I^d\big|_{k^* \beta} (f) =  \int_{\{ q \in \smop{Hom}(E, \mathbb R^d) / q_{|F} = k\}} f(q^* \beta ) dq.
$$
We recall that for all $A$ in $S_+^2 E^*$, for all $B$ in $S_+^2 E$ and for $\phi_B (A) := \exp (- \tr (AB))$ we have:
$$I^d\big|_C  (\phi_B ) = \pi^{(n-m)d/2} \exp ( - \mop{tr} (C. B^{F^*})).(\det B_{F^\bot})^{-d/2}.$$
We recall \cite{PE} the construction of the $D$-dimensional integral with parameters $(I^D)_{D \in \mathbb C} : \Cal S ( \bar S_+^2 E^*) \longrightarrow \Cal O (\mathbb C , \Cal S ( \bar S_+^2 F^*) )$ defined by:
$$ I^D \big|_C (\phi_B ) = \pi^{(n-m)D/2} \exp ( - \mop{tr} (C. B^{F^*})).(\det B_{F^\bot})^{-D/2}.$$
We define then the integral $I^D_{\Gamma, \delta} : \Cal S ( \bar S_+^2 E^*) \longrightarrow  \Cal S ( \bar S_+^2 F^*)$ for all $C \in S_+^2 F^*$ and $f \in \Cal S ( \bar S_+^2 E^*)$, by:
$$
I^D_{\Gamma, \delta} (f) (C) := I^D \big|_C (f),
$$
which is a holomorphic function in $D$.\\
We show that for all graphs $\Gamma$, $\gamma$ and $\delta$ such that $\delta \subseteq \gamma \subseteq \Gamma$ we have:
$$I^D_{\Gamma, \gamma} = I^D_{\Gamma, \delta} \circ I^D_{\Gamma/\delta, \Gamma/\gamma}.$$
We denote by $\Cal F (\bar S_+^2 E^*)$ and $\Cal F (S^2 E^*)$ respectively the two spaces of Feynman type functions on $\bar S_+^2 E^*$ and $S^2 E^*$, which are particular rational functions on $\bar S_+^2 E^*$ and $ S^2 E^*$ without real poles (see Definition \ref{fonc fey}), and by $\wt {\Cal F} (\bar S_+^2 E^*)$ the space of functions on $\mathbb C \times \bar S_+^2 E^*$, meromorphic in the first variable, equal to $I_{E', E}^D (g)$ for some function $g \in \Cal F (\bar S_+^2{E'}^*)$, where $E'$ is a vector space containing $E$. The integral $\wt I^D_{\Gamma, \delta} : \wt {\Cal F} (\bar S_+^2 E^*) \longrightarrow \wt {\Cal F} (\bar S_+^2 F^*)$ is defined by: 
$$\wt I^D_{\Gamma, \delta} (f)(C) :=  I^D \big|_C (f),$$
and extends to a meromorphic function of $D$.\\
We denote by $\mop{res} \Gamma$ the residue of the graph $\Gamma$. The Feynman rules are defined for $U = \Cal E (\mop{res} \Gamma)$ by:
$$\wt I_{\Gamma, \Gamma} \big(\varphi (\Gamma)\big) \in \wt {\Cal F} (\bar S_+^2 U^*),$$
where $\varphi (\Gamma)$ is the Feynman amplitude \eqref{propa} in \S\ref{des} below.\\
Let $G$ be the group of characters of $\Cal D_\Cal T$ with values in $\Cal A := \mop {End} \wt {\Cal B} ([z^{-1} , z]])$, equipped by the minimal renormalization scheme:. We show that every element $\varphi$ of $G$ has a unique Birkhoff decomposition compatible with the renormalization scheme chosen:
$$\varphi = \varphi _-^{\divideontimes-1} \divideontimes \varphi_+,$$
with:
$$
\varphi_- (\bar\Gamma,\bar\gamma) =  - P \Big( \varphi (\bar\Gamma,\bar\gamma) + \sum_{\substack{\bar\delta \subsetneq\; \bar\gamma \\ \bar\gamma / \bar\delta \; \in \Cal T}} \varphi (\bar\Gamma/\bar\delta, \bar\gamma/\bar\delta)  \circ \varphi_-(\bar\Gamma, \bar\delta) \Big), 
$$
$$
\varphi_+ (\bar\Gamma,\bar\gamma) = (I- P) \Big( \varphi (\bar\Gamma,\bar\gamma) + \sum_{\substack{\bar\delta \subsetneq\; \bar\gamma \\ \bar\gamma / \bar\delta \; \in \Cal T}} \varphi (\bar\Gamma/\bar\delta, \bar\gamma/\bar\delta)  \circ \varphi_-(\bar\Gamma, \bar\delta) \Big),
$$

and where $P : \Cal A \twoheadrightarrow \Cal A_-$ is the projection parallel to $\Cal A_+$. These formulas constitute the algebraic frame of Smirnov's approach \cite[\S 8.2]{Sm}. The regularized Feynman rules then define an element $\wt I$ of $G (k[z^{-1} , z]])$, 
$$\wt I : \wt{\Cal D}_{\Cal T} \longrightarrow \Cal A \;\;\;\; (\Gamma, \gamma) \longmapsto \wt I (\Gamma, \gamma) :=  \wt I_{\Gamma, \gamma}^D.$$
$$\text{In the Birkhoff decomposition : }\;\;\;\;\;\;\;\wt I = \wt I _-^{\divideontimes-1} \divideontimes \wt I_+,$$
the component $\wt I _-$ is the character of the counterterms, and the renormalized character $\wt I _+$ is evaluated at $D =d$.

\noindent
{\bf Acknowledgements:} I would like to thank Dominique Manchon for support and advice. Research supported by projet CMCU Utique 12G1502.
\section{Feynman graphs}
\subsection{Basic definitions}
A Feynman graph is a graph with a finite number of vertices and edges, which can be internal or external. An internal edge is an edge connected at both ends to a vertex, an external edge is an edge with one open end, the other end being connected to a vertex. The edges are obtained by using half-edges.\\
More precisely, let us consider two finite sets $\Cal V$ and $\Cal E$. A graph $\Gamma$ with $\Cal V$ (resp. $\Cal E$) as set of vertices (resp. half-edges) is defined as follows: let $\sigma : \Cal E \longrightarrow \Cal E $ be an involution and $\partial : \Cal E  \longrightarrow \Cal V$. For any vertex $v\in \Cal V$ we denote by $st(v) = \{e \in \Cal E / \partial (e) = v \}$ the set of half-edges adjacent to $v$. The fixed points of $\sigma$ are the \textsl {external edges} and the \textsl {internal edges} are given by the pairs $\{e , \sigma (e)\}$ for $e \neq \sigma (e)$. The graph $\Gamma$ associated to these data is obtained by attaching half-edges $e\in st(v)$ to any vertex $v\in\Cal V$, and joining the two half-edges $e$ and $\sigma(e)$ if $\sigma(e)\not =e$.\\

Several types of half-edges will be considered later on: the set $\Cal E$ is partitioned into several pieces $\Cal E_i$. In that case we ask that the involution $\sigma$ respects the different types of half-edges, i.e. $\sigma(\Cal E_i)\subset \Cal E_i$.\\

We denote by $\Cal I(\Gamma)$ the set of internal edges and by $\mop{Ext}(\Gamma)$ the set of external edges. The \textsl {loop number} of a graph $\Gamma$ is given by: 
$$L(\Gamma) =  \left|\Cal I (\Gamma)\right| - \left|\Cal V (\Gamma)\right|  + \left|\pi_0 (\Gamma)\right|,$$
where $\pi_0 (\Gamma)$ is the set of connected components of $\Gamma$.\\

A one-particle irreducible graph (in short, $1PI$ graph) is a connected graph which remains connected when we cut any internal edge. A disconnected graph is said to be locally $1PI$ if any of its connected components is $1PI$.

A covering subgraph of $\Gamma$ is a Feynman graph $\gamma$ (not necessarily connected), obtained from $\Gamma$ by cutting internal edges. In other words:
\begin{enumerate}
\item $ \Cal V (\gamma) = \Cal V (\Gamma)$. 
\item $ \Cal E (\gamma) =  \Cal E (\Gamma)$.
\item $ \sigma_\Gamma (e) = e \Longrightarrow \sigma_\gamma (e) = e$.
\item If $ \sigma_\gamma (e) \neq  \sigma_\Gamma (e) \;\; \text{then} \;\; \sigma_\gamma (e) = e \;\; \text{and} \;\; \sigma_\gamma (\sigma_\Gamma (e)) = \sigma_\Gamma (e)$.
\end{enumerate}

For any covering subgraph $\gamma$, the contracted graph $\Gamma / \gamma$ is defined by shrinking all connected components of $\gamma$ inside $\Gamma$ onto a point.  

The residue of the graph $\Gamma$, denoted by $\mop{res} \Gamma$, is the contracted graph $\Gamma / \Gamma$.

The skeleton of a graph $\Gamma$ denoted by $\mop {sk} \Gamma$ is a graph obtained by cutting all internal edges.
\subsection {Quantum field theory and specified graphs } 
We will work inside a physical theory $\Cal T$, ($\varphi^3 , \; \varphi^4$, QED, QCD  etc). The particular form of the Lagrangian leads to consider certain types of vertices and edges. A difficulty appears: the type of half-edges of $st (v)$ is not sufficient to determine the type of the vertex $v$. We denote by $\Cal E (\Cal T)$ the set of possible types of half-edges and by $\Cal V(\Cal T)$ the set of possible types of vertices.
\begin{ex}
$\Cal E ( \varphi^3) = \{ \graphearetenor \} \;\;\;, \;\;\; \Cal E ( QED) = \{ \graphearetenor  , \graphearetecur \}$.\\
${\Cal V}(\varphi^3) = \{ \graphearetenorcroixzero , \graphearetenorcroixun , \graphev  \} \;\;\;, \;\;\;{\Cal V}(QED) = \{ \graphearetenorcroixzero , \graphearetenorcroixun , \graphevcur , \graphearetecurcroixun \}$. 
\end{ex}
\begin{definition} \label{df1}
A specified graph of theory $\Cal T$ is a couple $(\Gamma , \underline{i})$ where:
\begin{enumerate}
\item $\Gamma$ is a locally $1PI$ superficially divergent graph with half-edges and vertices of the type prescribed in $\Cal T$. 
\item $\underline{i} : \pi_0(\Gamma) \longrightarrow \mathbb N$, the values of $\underline {i}(\gamma)$ being prescribed by the possible types of vertex obtained by contracting the connected component $\gamma$ on a point.
\end{enumerate}
We will say that $(\gamma , \underline{j})$ is a specified covering subgraph of $(\Gamma , \underline{i})$, $\big( (\gamma , \underline{j}) \subset (\Gamma , \underline{i}) \big)$  if:
\begin{enumerate}
\item $\gamma$ is a covering subgraph of $\Gamma$.
\item if $\gamma_0$ is a full connected component of $\gamma$, i.e if $\gamma_0$ is also a connected component of $\Gamma$, then $\underline{j} (\gamma_0) = \underline{i} (\gamma_0)$.
\end{enumerate}
\end{definition}
\begin{remark}
Sometimes we denote by $\bar \Gamma = (\Gamma , \underline{i})$ the specified graph, and we will write $\bar \gamma \subset \bar \Gamma$ for $(\gamma , \underline{j}) \subset (\Gamma , \underline{i})$.
\end{remark}
\begin{definition}
Let be $(\gamma , \underline{j}) \subset (\Gamma , \underline{i})$. The contracted specified subgraph is written:
$$\bar\Gamma / \bar \gamma = (\Gamma/\bar\gamma , \underline{i}),$$
where $\bar\Gamma/\bar\gamma$ is obtained by contracting each connected component of $\gamma$ on a point, and specifying the vertex obtained with $\underline {j}$.
\end{definition}
\begin{remark}
The specification $\underline{i}$ is the same for the graph $\bar\Gamma$ and the contracted graph $\bar\Gamma / \bar \gamma$. 
\end{remark}

\subsection {External structures}
Let $d$ be an integer $\geq 1$ (the dimension). For any half-edge $e$ of $\Gamma$ we denote by $p_e \in \mathbb R^d$ the corresponding moment. The momenta space of graph $\Gamma$ is defined by:
$$ W_\Gamma = \{ p :\Cal E (\Gamma) \longrightarrow \mathbb R^d , \sum_{e \in st(v)} p_e = 0 \;\;\forall v \in \Cal V (\Gamma) ,\;\; p_e + p_{\sigma(e)} = 0\;\; \forall e \in \Cal E (\Gamma) \;\;/ \;\; e \neq\sigma(e) \}.$$
The space of external momenta of $\Gamma$ is nothing but $W_{\smop{res} \Gamma}$. 
\section {Construction of Feynman graphs Hopf algebra}
\subsection {The bialgebra of specified graphs}
Let $\wt {\Cal H}_{\Cal T}$ be the vector space generated by the specified superficially divergent Feynman graphs of a field theory $\Cal T$. The product is given by the concatenation, the unit $\un$ is identified with the empty graph and the coproduct is defined by:
\begin{eqnarray*}
\Delta (\bar\Gamma ) = \sum_{\substack{\bar\gamma \subseteq \bar\Gamma \\ \bar\Gamma / \bar\gamma \in \Cal T }} \bar \gamma \otimes \bar\Gamma / \bar\gamma,
\end{eqnarray*}
where the sum runs over all locally $1PI$ specified covering subgraphs $ \bar\gamma = (\gamma,\underline{j})$ of $ \bar\Gamma = (\Gamma,\underline{i})$, such that the contracted subgraph $(\Gamma/(\gamma,\underline{j}) , \underline{i})$ is in the theory $\Cal T$.
\begin{remark}
The condition $\bar\Gamma / \bar\gamma \in \Cal T$ is crucial, and means also that $\bar\gamma$ is a ''superficially divergent'' subgraph. For example, in $\varphi^3$, for :
$$\Gamma = \graphedeltasansindice  \;\;\text{ and } \;\; \gamma = \graphedeltacontratddd \graphev \graphev,$$
gives $\Gamma / \gamma = \graphedeltacontratdddd$ by contraction, which must be eliminated because of the tetravalent vertex.  
\end{remark}
\begin{ex} \label{ex1}
In $\varphi^3$ Theory:
\begin{eqnarray*}
\Delta(\grapheexmdel , 0) &=& (\grapheexmdel , 0) \otimes \graphearetenorcroixzero  + \graphev \graphev \graphev \graphev \otimes (\grapheexmdel , 0)\\
&&\\
&+& (\graphev \graphev \graphecerc , 0)\otimes (\graphecerccroizero , 0)\\
&&\\
&+&  (\graphev \graphev \graphecerc , 1)\otimes (\graphecerccroiun , 0).
\end{eqnarray*}
\textnormal {In QED:}
\begin{eqnarray*}
\Delta (\qedj , 1) &=& \graphevcur \graphevcur \graphevcur \graphevcur\otimes (\qedj , 1)\\
&&\\
&+& (\qedj , 1) \otimes \graphearetenorcroixun  + (\qeda \graphevcur \graphevcur, 0) \otimes (\qedcoiz , 1)\\
&&\\ 
&+&(\qeda \graphevcur \graphevcur, 1) \otimes (\qedcoiu , 1).
\end{eqnarray*}
\end{ex}
\begin{theorem} \cite{DMB}
The coproduct $\Delta$ is coassociative. 
\end{theorem}
\subsection {The Hopf algebra of specified graphs}
The Hopf algebra ${\Cal H}_{\Cal T}$ is given by identifying all elements of degree zero (the residues) to unit $\un$:
\begin{equation}
{\Cal H}_{\Cal T} = \wt{\Cal H}_{\Cal T} / \Cal J
\end{equation}
where $\Cal J$ is the ideal generated by the elements $\un - \mop{res} \bar\Gamma$ where $\bar\Gamma$ is an $1PI$ specified graph. One immediately checks that $\Cal J$ is a bi-ideal. ${\Cal H}_{\Cal T}$ is a connected graded bialgebra, it is therefore a connected graded Hopf algebra. The coproduct then becomes:
\begin {equation}
\Delta (\bar\Gamma ) = \un \otimes \bar\Gamma + \bar\Gamma \otimes \un + \sum_{\substack{\bar\gamma \hbox{ \sevenrm proper subgraph of}\; \bar\Gamma \\ \hbox{ \sevenrm loc 1PI.}\; \bar\Gamma / \bar\gamma \in \Cal T }} \bar\gamma \otimes \bar\Gamma /\bar\gamma,
\end{equation}
\section {Doubling the bialgebra of specified graphs}
Let $\wt{\Cal D}_{\Cal T}$ be the vector space spanned by the pairs $(\bar\Gamma, \bar\gamma)$ of locally $1PI$ specified graphs, with $\bar\gamma \subset \bar\Gamma$ and $\bar\Gamma / \bar\gamma \in \wt {\Cal H}_{\Cal T}$. This is the free commutative algebra generated by the corresponding connected objects. The product is again given by juxtaposition, and the coproduct is defined as follows: 
\begin{equation}
\Delta (\bar\Gamma, \bar\gamma ) = \sum_{\substack{\bar\delta \subseteq \bar\gamma \\ \bar\gamma / \bar\delta \in \Cal T }} ( \bar \Gamma, \bar\delta )  \otimes ( \bar\Gamma / \bar\delta, \bar\gamma / \bar\delta)
\end{equation}
\begin{proposition} 
$({\wt{\Cal D}}_{\Cal T}, m, \Delta, u,\varepsilon)$ is a graded bialgebra, and 
\begin{eqnarray*}
P_2: {\wt{\Cal D}}_{\Cal T} &\longrightarrow& {\wt{\Cal H}}_{\Cal T} \\ (\bar\Gamma, \bar\gamma ) &\longmapsto& \bar\gamma
\end{eqnarray*} 
is a bialgebra morphism. 
\end{proposition}
\begin{proof}
The unit is the pair $(\textbf{1}, \textbf{1})$, where $\textbf{1}$ is the empty graph, co-unit is given by $\varepsilon (\bar\Gamma, \bar\gamma ) = \varepsilon (\bar\gamma )$, the grading is given by the loop number of the subgraph:
\begin{equation}
 |(\bar\Gamma, \bar\gamma )| = |\bar\gamma|.
\end{equation}
Let us now check coassociativity:
\begin{eqnarray*}
(\Delta \otimes id)\Delta (\bar\Gamma, \bar\gamma ) &=& (\Delta \otimes id) \Big( \sum_{\substack{\bar\delta \subseteq \bar\gamma \\ \bar\gamma / \bar\delta \in \Cal T }} ( \bar \Gamma, \bar\delta )  \otimes ( \bar\Gamma / \bar\delta, \bar\gamma / \bar\delta) \Big)\\
&=& \sum_{\substack{\bar\sigma \subseteq \bar\delta \subseteq \bar\gamma \_ \bar\delta / \bar\sigma \;,\;\bar\gamma / \bar\delta \in \Cal T }} ( \bar \Gamma, \bar\sigma )  \otimes ( \bar\Gamma / \bar\sigma, \bar\delta / \bar\sigma) \otimes ( \bar\Gamma / \bar\delta, \bar\gamma / \bar\delta),
\end{eqnarray*}
whereas,
\begin{eqnarray*}
(id \otimes \Delta)\Delta (\bar\Gamma, \bar\gamma ) &=& (id \otimes \Delta) \Big( \sum_{\substack{\bar\sigma \subseteq \bar\gamma \\ \bar\gamma / \bar\sigma \in \Cal T }} ( \bar \Gamma, \bar\sigma )  \otimes ( \bar\Gamma / \bar\sigma, \bar\gamma / \bar\sigma) \Big)\\  
&=& \sum_{\substack{\bar\sigma \subseteq \bar\gamma \\ \bar\gamma / \bar\sigma \in \Cal T }} \sum_{\substack{\wt\delta \subseteq \bar\gamma / \bar\sigma \\ \bar\gamma / \bar\sigma / \wt\delta \in \Cal T }} (\bar \Gamma, \bar\sigma) \otimes (\bar\Gamma / \bar\sigma, \wt\delta) \otimes (\bar\Gamma / \bar\sigma \Big/\wt\delta , \bar\gamma / \bar\sigma \Big/ \wt\delta)\\ 
&=& \sum_{\substack{\bar\sigma \subseteq \bar\delta \subseteq \bar\gamma \\ \bar\gamma / \bar\sigma \;,\;\bar\gamma / \bar\sigma \in \Cal T }} (\bar \Gamma, \bar\sigma )  \otimes ( \bar\Gamma / \bar\sigma, \bar\delta / \bar\sigma) \otimes ( \bar\Gamma / \bar\delta, \bar\gamma / \bar\delta).
\end{eqnarray*}
The conditions $\{\bar\gamma / \bar\delta \;,\;\bar\gamma / \bar\sigma \in \Cal T \}$ and $\{\bar\gamma / \bar\delta \;,\;\bar\delta / \bar\sigma \in \Cal T \}$ are equivalent, which proves coassociativity. Compatibility with product and grading are obvious. Finally $P_2$ is an algebra morphism, and we have:
\begin{eqnarray*}
\Delta \circ P_2 (\bar\Gamma, \bar\gamma ) &=& \Delta(\bar\gamma ) \\  
&=& \sum_{\substack{\bar\gamma \subseteq \bar\Gamma \\ \bar\Gamma / \bar\gamma \in \Cal T }} \bar \delta \otimes \bar\gamma / \bar\delta\\ 
&=& (P_2 \otimes P_2) \Delta (\bar\Gamma, \bar\gamma ).
\end{eqnarray*}
\end{proof}
\section {Feynman rules}
\subsection {Describing the integrand}\label{des}
Let $\Gamma$ be a Feynman graph. Every vertex $v$ comes with its coupling constant $g_v$ depending only on the vertex type of $v$. Every internal edge $(e_- e_+)$ (resp every external edge $e$) comes with  propagator $G_{e_- e_+} : \mathbb R^d \longrightarrow \mathbb C$ (resp $G_e$) which depends only on the type of edge. This is a rational function without real poles in the Euclidean case with mass. For example in $\varphi^3$ or $\varphi^4$ theory, all propagators (internal and external) are given by the same function $G$:
$$G (p) = \frac{1}{p^2 + m^2}.$$
Then $\varphi(\Gamma)$ is an element of $\Cal C^{\infty} (W_\Gamma ,\mathbb C)$ defined by:
\begin{equation}\label{propa}
 \varphi(\Gamma)(p) = \prod_{v\in \Cal V (\Gamma)} g_v \prod_{\{e,\sigma(e)\}, \sigma(e)\neq e} G_{e \sigma(e)} (p_e)\prod_{\sigma(e)= e} G_e (p_e).
\end{equation}
Note that if $G_{e \sigma(e)}$ is not an even function, we must orient the corresponding internal edge and we set $G_{\sigma(e) e} (p) = G_{e \sigma(e)} (- p)$.
\begin{ex} We consider the following graph in $\varphi^3$ theory:
\vskip-0.8cm 
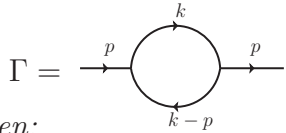
$$\hskip-3cm\Gamma = \grphopropg$$
The amplitude of $\Gamma$ can be written:
$$\varphi(\Gamma) = g^2 . \frac{1}{(p^2 + m^2)^2} . \frac{1}{k^2 + m^2}. \frac{1}{(k - p)^2 + m^2}.$$ 
\end{ex}
\subsection {Integrating internal momenta}
We put on $\mathbb R$ a normalized Lebesgue measure 
$$\dbar \xi = \frac {d\xi} {2\pi}$$ 
and we associate the normalized Lebesgue measure on $W_{\Gamma}$.\\
For all $q$ in $W_{\Gamma/\gamma}$ we set: 
\begin{equation}
I_{\Gamma, \gamma} \varphi (q) = \int_{ {\Cal F}^{-1}_{\Gamma,\gamma}(\{q\}) \subset W_{\Gamma}} \varphi (p) \dbar p.
\end{equation}
Hence:
$$ I_{\Gamma, \gamma} : \Cal C^{\infty} (W_\Gamma) \longrightarrow \Cal C^{\infty} (W_{\Gamma/\gamma} ),$$
if the integral converges. Integrating by stages, if $\delta \subset \gamma \subset \Gamma$, we have:
$$I_{\Gamma, \gamma} = I_{\Gamma/\delta, \gamma/\delta} \circ I_{\Gamma, \delta}$$
The Feynman rules are given by: 
$$ \Gamma \longmapsto I_{\Gamma, \Gamma} (\varphi (\Gamma)).$$
The problem is that these integrals are in general divergent.
\section {Bogoliubov's algorithm}
We denote by $\pi$ the projection on the counterterms. Here $\pi : \Cal C^{\infty} (W_\Gamma) \twoheadrightarrow \Cal C^{\infty} (W_\Gamma)$ is given by Taylor expansion at a certain order, or in dimensional regularization one often considers $\pi : \Cal C^{\infty} (W_\Gamma) [z^{-1} ,z]] \twoheadrightarrow z^{-1} \Cal C^{\infty} (W_\Gamma)[z^{-1}]$ (minimal scheme).\\
If $|\gamma| = 0$ we find that: $I_{\Gamma, \gamma} = Id$.\\
If $|\gamma| = 1$ we set: $ I_{\Gamma, \gamma}^{-} := - \pi \circ  I_{\Gamma, \gamma}.$\\
If $|\gamma| \geq 2$, we use our favorite recursive formulas:
\begin{equation}
 I_{\Gamma, \gamma}^{-} := - \pi \circ \Big(  I_{\Gamma, \gamma} + \sum_{\delta \subset \gamma}  I_{\Gamma/\delta, \gamma/\delta} \circ  I_{\Gamma, \delta}^{-} \Big),
\end{equation}
\begin{equation}
I_{\Gamma, \gamma}^{+} := (Id- \pi )\circ \Big( I_{\Gamma, \gamma} + \sum_{\delta \subset \gamma}  I_{\Gamma/\delta, \gamma/\delta} \circ  I_{\Gamma, \delta}^{-}\Big).
\end{equation}
This is very similar to the Connes-Kreimer formulas \cite{A.D2000}, but there are some differences: the commutative product of the target algebra is replaced by composition of operators. On the other hand, one can guess that the bialgebra $\wt {\Cal D}_{\Cal T}$ must play a role: this is indeed case. Moreover as $I_{\Gamma, \gamma} = Id$ for any $\gamma$ of degree zero, we can rather work with the Hopf algebra ${\Cal D}_{\Cal T} := \wt{\Cal D}_{\Cal T} / \Cal J$, where $\Cal J$ is the (bi-) ideal spanned by $\un - (\bar\Gamma,\gamma)$, for $|\gamma| = 0$.
\section {Adaptation of the Connes-Kreimer formalism}
We want to reinterpret Smirnov's approach \cite[\S 8]{Sm} into Connes-Kreimer formalism, using the Hopf algebra ${\Cal D}_{\Cal T}$ and a target algebra which is no longer commutative, because of composition of operators.  
\subsection{The unordered tensor product}
Let $A$ be a finite set, and let $V_j$ be a vector space for any $j\in A$. The product $\prod_{j\in A} V_{j}$ is defined by:
$$\prod_{j\in A} V_{j} := \{ v : A \longrightarrow \coprod_{j\in A} V_j \;, \; v(j) \; \in \;  V_{j} \;\forall \; j \; \in A \}.$$ 
The space $V := \bigotimes_{j\in A} V_{j}$ is then defined by the following universal property: for any vector space $E$ and for any multilinear map $F : \prod_{j\in A} V_{j} \longrightarrow E$, there exists a unique linear map $\bar F$ such that the following diagram is commutative:
\diagramme{
\xymatrix{
& \bigotimes_{j\in A} V_{j} \ar[d]^{\bar F} \\
 \prod_{j\in A} V_{j} \ar[ru]^{v \mapsto \bigotimes_{j\in A} v_j} \ar[r]_{F} & E }
 }
\begin{remark}
Let $\big( e_{\lambda} \big)_{\lambda \in {\Lambda_j}}$ be a basis of $V_{j}$. A basis of $\bigotimes_{j\in A} V_{j}$ is given by:
$$ \big( f_\mu = \bigotimes_{j\in A} e_{\mu(j)} \big)_{\mu \in  \Lambda},$$ 
where $ \Lambda = \prod_{j\in A} \Lambda_j =  \{ \mu : A \longrightarrow \coprod_{j\in A} \Lambda_j \;\; \text {such that}\;\;  \mu (j) \in \Lambda_j \}$.  
\end{remark}
\subsection {An algebra of Schwartz functions}\label{sch}
We introduce, for a connected graph $\Gamma$:
\begin{equation}
 V_\Gamma := \Cal S (W_\Gamma ).
\end{equation} 
For $\Gamma = \Gamma_1 \cdots \Gamma_r$ not connected, we put:
$$ V_\Gamma = \bigotimes_{j\in \{1, \ldots, r\}} V_{\Gamma_j},$$
where the space $V_\Gamma := \bigotimes_{j\in A} V_{\Gamma_j}$ is the unordered tensor product of the $V_{\Gamma_j}$'s and where the $\Gamma_j$'s are the connected components of $\Gamma$. Finally we consider\\
\begin{equation}\label{cab}
\Cal B := \prod_{\Gamma \in \Cal T} V_{\Gamma}.
\end{equation} 
Let use remark that the operators $I_{\Gamma, \gamma}$ can be understood as elements of $\mop {End} \Cal B$.
 
Any $a \in \mop {End} \Cal B$ is written as a block matrix whose coefficients are of the following form:
$$a_{\Gamma \Gamma'} : V_\Gamma \longrightarrow V_\Gamma'.$$
We equip $\mop {End} \Cal B$ with a product $\bullet$ defined by:
\begin{eqnarray*}
a_{\Gamma_1 \Gamma'_1} \bullet a_{\Gamma_2 \Gamma'_2} : V_{\Gamma_1 \Gamma_2} &\longrightarrow& V_{\Gamma'_1 \Gamma'_2} \\ \bigotimes_{j\in \{1, 2\}} v_j &\longmapsto& \bigotimes_{j\in \{1, 2\}} a_{\Gamma_j \Gamma'_j} (v_j)
\end{eqnarray*}  
This definition extends naturally to a commutative bilinear product on $\mop {End} \Cal B$:
\begin{equation}
(a \bullet b)_{\Gamma \Gamma'} = \sum_{\substack{\partial \delta = \Gamma\\ \partial' \delta' = \Gamma' }} a_{\partial \partial'} \bullet b_{\delta \delta'}.
\end{equation}
\noindent Considering four linear maps:
$$a_1 : V_{\Gamma_1} \longrightarrow V_{\Gamma'_1} ; \;\;\;\;\; b_1 : V_{\Gamma'_1} \longrightarrow V_{\Gamma''_1} ,$$
$$a_2 : V_{\Gamma_2} \longrightarrow V_{\Gamma'_2} ; \;\;\;\;\; b_2 : V_{\Gamma'_2} \longrightarrow V_{\Gamma''_2} ,$$
we have the following result:
\begin{proposition} \label{ct}
$$(b_1 \circ a_1) \bullet (b_2 \circ a_2) = (b_1 \bullet b_2) \circ (a_1 \bullet a_2).$$
\end{proposition}
\begin{proof}
\begin{eqnarray*}
(b_1 \circ a_1) \bullet (b_2 \circ a_2) ( \bigotimes_{j\in \{1, 2\}} v_j) &=& \bigotimes_{j\in \{1, 2\}} (b_j \circ a_j) (v_j)\\
\end{eqnarray*}
\begin{eqnarray*}
(b_1 \bullet b_2) \circ (a_1 \bullet a_2)( \bigotimes_{j\in \{1, 2\}} v_j) &=& b_1 \bullet b_2 \Big( \bigotimes_{j\in \{1, 2\}} a_j(v_j) \Big)\\
&=&  \bigotimes_{j\in \{1, 2\}} b_j ( a_j(v_j))\\
&=&\bigotimes_{j\in \{1, 2\}} (b_j \circ a_j) (v_j).
\end{eqnarray*}
\end{proof}
\subsection {Convolution product $\divideontimes$}
We denote by $\diamond$ the opposite of composition product in $\mop {End} \Cal B$. Then we define a convolution product $\divideontimes$ for all $\varphi$, $\psi \in {\Cal L} ({\Cal D}_{\Cal T}, \mop {End} \Cal B)$ by: 
\begin{equation}
 \varphi \divideontimes \psi :=  \diamond  ( \varphi \otimes \psi )  \Delta.
\end{equation}
In other words, for all specified graphs $\bar\gamma$, $\bar\Gamma$ such that $\bar\gamma \subset \bar\Gamma$ we have:
\begin{equation}
(\varphi \divideontimes \psi) (\bar\Gamma, \bar\gamma ) = \sum_{\substack{\bar\delta \subseteq \bar\gamma \\ \bar\gamma / \bar\delta \in \Cal T }}\psi( \bar\Gamma / \bar\delta, \bar\gamma / \bar\delta)\circ \varphi( \bar \Gamma, \bar\delta ).
\end{equation}
\begin{theorem}
Let $G$ be the set of morphisms of unitary algebras: ${\Cal D}_{\Cal T} \longrightarrow (\mop {End} \Cal B , \bullet)$. Equipped with the product $\divideontimes$, the set $G$ is a group. 
\end{theorem}
\begin{proof}
Associativity of $\divideontimes$ is immediate, it results from the associativity of $\diamond$ and coassociativity of $\Delta$. The identity element $E$ is defined by:
\begin{equation}
E(\bar\Gamma, \bar\gamma)= \left\lbrace
\begin{array}{lcl}
Id_{\Cal B}  \;\;\;\;\; \text{if} \;\; \bar\gamma \;\; \text{is a specified graph of degree zero}\\
0  \;\; \;\; \;\; \;\; \text{if not}.
\end{array}\right. 
\end{equation}
The inverse of an element $\varphi$ of $G$ is given by the following recursive formula:
\begin{equation}
\left\lbrace
\begin{array}{lcl}
\varphi^{\divideontimes -1} (\bar\Gamma, \bar\gamma) = \varphi(\bar\Gamma, \bar\gamma) = Id_{\Cal B} \;\;\;\; \text{if} \;\; |\gamma| = 0\\
&&\\
\varphi^{\divideontimes -1} (\bar\Gamma, \bar\gamma) = \sum_{\substack{\bar\delta \subseteq \bar\gamma \\ |\delta| \geq 1}} \varphi^{\divideontimes -1} (\bar\Gamma/\bar\delta, \bar\gamma/\bar\delta) \circ \varphi(\bar\Gamma, \bar\delta) \;\;\;\;\text{if} \;\; |\gamma| \geq 1.
\end{array}\right. 
\end{equation}
The fact that for all $\varphi$,  $\psi \in G$ we have $\varphi \divideontimes \psi \in G$ is a result of the compatibility between the products $\diamond$ and $\bullet$. Indeed:
\begin{eqnarray*}
(\varphi \divideontimes \psi) (\bar\Gamma \bar\Gamma', \bar\gamma\bar\gamma' ) &=&(\varphi \divideontimes \psi) (\Gamma \Gamma',\gamma\bar\gamma')\\
&=& \sum_{\substack{\bar\delta \subseteq \bar\gamma \bar\gamma'\\ (\bar\gamma\bar\gamma') / \bar\delta \in \Cal T}} \psi \big( \bar\Gamma\bar\Gamma' / \bar\delta, (\bar\gamma\bar\gamma') / \bar\delta\big)\circ \varphi \big(\bar \Gamma\bar \Gamma', \bar\delta \big)\\  
&=& \sum_{\substack{\bar\delta  \subset\; \bar\gamma \;,\;\bar\delta' \subset\; \bar\gamma' \\ \bar\gamma / \bar\delta ; \bar\gamma' / \bar\delta' \in \Cal T}} 
\psi\big(\bar\Gamma\bar\Gamma' / \bar\delta\bar\delta', \bar\gamma\bar\gamma' / \bar\delta\bar\delta'\big)\circ \varphi \big(\bar \Gamma\bar \Gamma', \bar\delta\bar\delta'\big)\\
&=& \sum_{\substack{\bar\delta  \subset\; \bar\gamma \;,\;\bar\delta' \subset\; \bar\gamma' \\ \bar\gamma / \bar\delta ; \bar\gamma' / \bar\delta' \in \Cal T}} 
\Big( \psi( \bar\Gamma / \bar\delta, \bar\gamma / \bar\delta)\bullet \psi( \bar\Gamma' / \bar\delta', \bar\gamma' / \bar\delta')\Big)\circ \Big(\varphi( \bar \Gamma\bar , \bar\delta)\bullet \varphi(\bar \Gamma', \bar\delta')\Big)\\
&=& \sum_{\substack{\bar\delta  \subset\; \bar\gamma \;,\;\bar\delta' \subset\; \bar\gamma' \\ \bar\gamma / \bar\delta ; \bar\gamma' / \bar\delta' \in \Cal T}} 
\Big( \psi( \bar\Gamma / \bar\delta, \bar\gamma / \bar\delta)\circ \varphi( \bar \Gamma\bar , \bar\delta)\Big)\bullet \Big(\psi( \bar\Gamma' / \bar\delta', \bar\gamma' / \bar\delta')\circ \varphi(\bar \Gamma', \bar\delta')\Big)\\
&=&(\varphi \divideontimes \psi) (\bar\Gamma, \bar\gamma ) \bullet (\varphi \divideontimes \psi) (\bar\Gamma', \bar\gamma').
\end{eqnarray*}
\end{proof}
\subsection {Taylor expansions} 
We adapt here a construction from \cite[\S 9]{EGP} also used by \cite[\S 3.7]{EP}, (see also \cite{KP, KP2}, and also \cite{DMB}). 
\begin{definition}
Let $\Cal B$ be the commutative algebra defined by \eqref{cab}. For $m \in \mathbb N$ the order $m$ Taylor expansion operator is:
\begin{equation}
P_m \in \mop{End}(\Cal B),  \;\;\;\; P_m f  (v) := \sum_{|\beta| \leq m} \frac{v^\beta}{\beta!}\partial_0^\beta f,
\end{equation}
where $\beta = (\beta_1, ..., \beta_n) \in {\mathbb N}^n$ with the usual notations $\beta \leq \alpha$ iff $\beta_i \leq \alpha_i$ for all $i$, $|\beta|:= \beta_1 + ...+ \beta_n$ as well as 
$$v^\beta = \prod_{1 \leq k \leq n} v_k^{\beta_k} , \;\;\;\;\;\;\;\;\; \beta! := \prod_{1 \leq k \leq n} \beta_k! ,\;\;\; \;\;\;\;\;\; \partial_0^\beta := \prod_{1 \leq k \leq n} \frac{\partial^{\beta_k}}{\partial v_k^{\beta_k}}_{|v_k =0}.$$  	
\end{definition}
Let $\varphi$ be an element of $G$ lifted to an algebra morphism from $\wt {\Cal D}_{\Cal T}$ into $(\mop {End} {\Cal B}, \bullet)$. We define a Birkhoff decomposition:
\begin{equation}
\varphi = \varphi _-^{\divideontimes-1} \divideontimes \varphi_+.
\end{equation}
The components $\varphi_+$ and $\varphi_-$ are given by simple recursive formulas: for any $(\bar\Gamma, \bar\gamma)$ of degree zero (i.e $\gamma$ without internal edges) we put: $\varphi _-(\bar\Gamma, \bar\gamma) = \varphi _+(\bar\Gamma, \bar\gamma) = \varphi (\bar\Gamma, \bar\gamma)= Id_{\Cal B}$. If we assume that $\varphi_- (\bar\Gamma, \bar\gamma)$ and $\varphi_+ (\bar\Gamma, \bar\gamma)$ are known for $(\bar\Gamma, \bar\gamma)$ of degree $k \leq m-1$, we have then for any pairs of	specified graphs $(\bar\Gamma, \bar\gamma)$ of degree $m$:
\begin{equation}
\varphi_- (\bar\Gamma,\bar\gamma) = - P_m \Big( \varphi (\bar\Gamma,\bar\gamma) + \sum_{\substack{\bar\delta \subsetneq\; \bar\gamma \\ \bar\gamma / \bar\delta \; \in \Cal T}} \varphi (\bar\Gamma/\bar\delta, \bar\gamma/\bar\delta)  \circ \varphi_-(\bar\Gamma, \bar\delta) \Big) \label{deb1}
\end{equation}
\begin{equation}
\varphi_+ (\bar\Gamma,\bar\gamma) = (I- P_m) \Big( \varphi (\bar\Gamma,\bar\gamma) + \sum_{\substack{\bar\delta \subsetneq\; \bar\gamma \\ \bar\gamma / \bar\delta \; \in \Cal T}} \varphi (\bar\Gamma/\bar\delta, \bar\gamma/\bar\delta)  \circ \varphi_-(\bar\Gamma, \bar\delta) \Big). \label{debb1}
\end{equation}
The operators $P_m$ form a \textsl{Rota--Baxter family} in the sense of K. Ebrahimi-Fard, J. Gracia-Bondia and F. Patras \cite[Proposition 9.1, Proposition 9.2]{EGP}. For all graph $\Gamma$, for all $f, g \in V_\Gamma$ and for any $s, t \in \mathbb N$ we have:
\begin{equation}
 ( P_s f ) (P_t g) = P_{s+t} [ (P_s f)g + f (P_t g) - fg]. \label{debb}
\end{equation} 
\begin{theorem}\label{dt}
Let $\varphi$ be a character of $\wt{\Cal D}_{\Cal T}$ with values in the unitary commutative algebra $\Cal B$. Further let $P_.: \mathbb N \longrightarrow \mop{End}(\Cal B)$ be an indexed renormalization scheme, that is a family $(P_t)_{t \in \mathbb N}$ of endomorphisms such that: 
\begin{equation}\label{defam}
 \mu \circ (P_s \otimes P_t) = P_{s+t} \circ \mu \circ [P_s \otimes Id + Id \otimes P_t - Id \otimes Id],
\end{equation}  
for all $s, t \in \mathbb N$. Then the two maps $\varphi_-$ and $\varphi_+$ defined by \eqref{deb1} and \eqref{debb1} are two characters of $(\Cal B, \bullet)$.
\end{theorem}
\begin{proof}
We will just prove that $\varphi_-$ is a character. Then  $\varphi_+ = \varphi_- \circledast \varphi$ is also a character. For $(\bar\Gamma,\bar\gamma), (\bar\Gamma',\bar\gamma') \in \mop{ker} \varepsilon$, we write $\varphi_- (\bar\Gamma,\bar\gamma) = - P_{|(\bar\Gamma,\bar\gamma)|} (\bar\varphi (\bar\Gamma,\bar\gamma) )$ where: 
$$\bar\varphi (\bar\Gamma,\bar\gamma) = \varphi (\bar\Gamma,\bar\gamma) + \sum_{\substack{\bar\delta \subsetneq\; \bar\gamma \\ \bar\gamma / \bar\delta \; \in \Cal T}} \varphi (\bar\Gamma/\bar\delta, \bar\gamma/\bar\delta)  \circ \varphi_-(\bar\Gamma, \bar\delta).$$
For proving this theorem we use the formulas \eqref{deb1} and \eqref{defam}.
\begin{eqnarray*}
\varphi_- (\bar\Gamma\bar\Gamma',\bar\gamma\bar\gamma') &=& - P_{|(\bar\Gamma\bar\Gamma',\bar\gamma\bar\gamma')|} \Big( \varphi (\bar\Gamma\bar\Gamma',\bar\gamma\bar\gamma') + \sum_{\substack{\bar\delta\bar\delta' \subsetneq \bar\gamma\bar\gamma' \\ \bar\gamma\bar\gamma' / \bar\delta\bar\delta' \; \in \Cal T}} \varphi (\bar\Gamma \bar\Gamma'/ \bar\delta\bar\delta', \bar\gamma \bar\gamma'/\bar\delta \bar\delta')  \circ \varphi_-(\bar\Gamma\bar\Gamma', \bar\delta\bar\delta') \Big)\\ 
&=& - P_{|(\bar\Gamma\bar\Gamma',\bar\gamma\bar\gamma')|} \Big[\varphi (\bar\Gamma,\bar\gamma)\bullet \varphi (\bar\Gamma',\bar\gamma') + \varphi_- (\bar\Gamma,\bar\gamma) \bullet  \varphi (\bar\Gamma',\bar\gamma')+ \varphi (\bar\Gamma ,\bar\gamma)\bullet \varphi_- (\bar\Gamma',\bar\gamma')\\
&+& \sum_{\substack{\bar\delta \subsetneq \bar\gamma \\ \bar\gamma / \bar\delta \; \in \Cal T}} \left(\varphi (\bar\Gamma / \bar\delta, \bar\gamma /\bar\delta) \circ \varphi_-(\bar\Gamma, \bar\delta) \right) \bullet \left( \varphi_-(\bar\Gamma', \bar\gamma') + \varphi (\bar\Gamma', \bar\gamma') \right)\\
&+& \sum_{\substack{\bar\delta' \subsetneq \bar\gamma' \\ \bar\gamma' / \bar\delta' \; \in \Cal T}} \left(\varphi (\bar\Gamma' / \bar\delta', \bar\gamma' /\bar\delta') \circ \varphi_-(\bar\Gamma', \bar\delta') \right) \bullet \left( \varphi_-(\bar\Gamma, \bar\gamma) + \varphi (\bar\Gamma, \bar\gamma) \right)\\
&+&\sum_{\substack{\bar\delta \subsetneq\; \bar\gamma \;,\; \bar\delta' \subsetneq\; \bar\gamma' \\ \bar\gamma / \bar\delta \;;\;\bar\gamma' / \bar\delta' \in \Cal T}}
\Big(\varphi (\bar\Gamma / \bar\delta, \bar\gamma/\bar\delta)  \circ \varphi_-(\bar\Gamma, \bar\delta)\Big) \bullet \Big(\varphi (\bar\Gamma'/ \bar\delta', \bar\gamma'/\bar\delta')  \circ \varphi_-(\bar\Gamma', \bar\delta')\Big)\Big]\\
&=& - P_{|(\bar\Gamma, \bar\gamma)|+ |(\bar\Gamma', \bar\gamma')|} \Big[ \Big( \varphi(\bar\Gamma, \bar\gamma) + \sum_{\substack{\bar\delta \subsetneq\; \bar\gamma \\ \bar\gamma / \bar\delta \; \in \Cal T}} \varphi (\bar\Gamma / \bar\delta, \bar\gamma /\bar\delta) \circ \varphi_-(\bar\Gamma, \bar\delta)\Big)\\
&& \hskip 4cm \bullet \Big( \varphi(\bar\Gamma', \bar\gamma') + \sum_{\substack{\bar\delta' \subsetneq\; \bar\gamma' \\ \bar\gamma' / \bar\delta' \; \in \Cal T}} \varphi (\bar\Gamma' / \bar\delta', \bar\gamma' /\bar\delta') \circ \varphi_-(\bar\Gamma', \bar\delta')\Big)\\
&+& \varphi_-(\bar\Gamma', \bar\gamma') \bullet \Big( \varphi(\bar\Gamma, \bar\gamma) + \sum_{\substack{\bar\delta \subsetneq\; \bar\gamma \\ \bar\gamma / \bar\delta \; \in \Cal T}} \varphi (\bar\Gamma / \bar\delta, \bar\gamma /\bar\delta) \circ \varphi_-(\bar\Gamma, \bar\delta)\Big)\\
&+&\varphi_-(\bar\Gamma, \bar\gamma) \bullet \Big( \varphi(\bar\Gamma', \bar\gamma') + \sum_{\substack{\bar\delta' \subsetneq\; \bar\gamma' \\ \bar\gamma' / \bar\delta' \; \in \Cal T}} \varphi (\Gamma' / \delta', \gamma' /\delta') \circ \varphi_-(\bar\Gamma', \bar\delta')\Big) \Big]\\
&=&- P_{|(\bar\Gamma, \bar\gamma)|+ |(\bar\Gamma', \bar\gamma')|} \Big[\bar\varphi (\bar\Gamma,\bar\gamma)\bullet \bar\varphi (\bar\Gamma',\bar\gamma') - P_{|(\bar\Gamma, \bar\gamma)|} (\bar\varphi (\bar\Gamma,\bar\gamma))\bullet \bar\varphi (\bar\Gamma',\bar\gamma') \\
&&\hskip 4cm  -  P_{|(\bar\Gamma', \bar\gamma')|}(\bar\varphi (\bar\Gamma',\bar\gamma'))\bullet \bar\varphi (\bar\Gamma,\bar\gamma) \Big]\\
&=& P_{|(\bar\Gamma, \bar\gamma)|+ |(\bar\Gamma', \bar\gamma')|} \Big[ P_{|(\bar\Gamma', \bar\gamma')|} (\bar\varphi (\bar\Gamma',\bar\gamma'))\bullet \bar\varphi (\bar\Gamma,\bar\gamma) + P_{|(\bar\Gamma, \bar\gamma)|}(\bar\varphi (\bar\Gamma,\bar\gamma))\bullet\bar\varphi (\bar\Gamma',\bar\gamma') \\
&&\hskip 4cm  - \bar\varphi (\bar\Gamma,\bar\gamma) \bullet \bar\varphi (\bar\Gamma',\bar\gamma')\Big]\\
&=&  P_{|(\bar\Gamma, \bar\gamma)|} \big(\bar\varphi (\bar\Gamma,\bar\gamma)\big)\bullet  P_{|(\bar\Gamma', \bar\gamma')|}\big(\bar\varphi (\bar\Gamma',\bar\gamma')\big)\\
&=& \varphi_- (\bar\Gamma,\bar\gamma) \bullet \varphi_- (\bar\Gamma',\bar\gamma').
\end{eqnarray*}
\end{proof}
\section {Dimensional regularization} 
The problem of perturbative renormalization theory is to give a meaning to certain divergent integrals arising from Feynman graphs. The analytical difficulty is to regularize the occurring divergent integrals, associating to each Feynman graph some finite value indexed by a parameter. The procedure we use to regularize an integral is called dimensional regularization. The basic idea behind dimensional regularization consists in writing down the divergent integrals that we have to regularize in such a way that the dimension of the physical space-time $d$ becomes an complex parameter $D$. We follow the approach of Pavel Etingof \cite{PE} and Ralf Meyer \cite {RM}. 
\subsection {General idea}
Let $V$ the $d$-dimensional space-time. Consider a Feynman graph with $m$ external edges and corresponding momenta $q_1, \ldots, q_m \in V$, and with $n - m$ loops $(n \geq m)$. The amplitude of this graph can be written as:
\begin{equation}
I_{( q_1  \ldots q_m )} (f)  :=  \int_{ V^{n-m}} f( q_1, \ldots, q_n ) dq_{m+1}  \ldots dq_n,
\end{equation}
where we assume at first that $f$ is a Schwartz function. The dimension $d$ of space-time has become an external parameter which can be replaced by any complex number $D$. We obtain the following results: first, for $\Re D > n+1$ and $\Re D < 0$ the integral $I^D (f)$ exists. Moreover, this function admits a holomorphic extension with respect to $D$ on $\mathbb C$. If $f$ is a Feynman type rational function (see Definition \ref{fonc fey} below), $I^D (f)$ is defined on a half-plane but admits a meromorphic extension.\\
Once we have a meromorphic function $I^D (f)$, the regularization value of the integral is obtained by minimal subtraction at the physical dimension $D = d$. We consider the Laurent series around $d$:
\begin{equation}
 I^D (f) =  \sum_{n \in \mathbb Z} a_n .( D - d )^n.
\end{equation}
Then the counterterm is given by:
\begin{equation}
 I^{D}_- (f) :=  \sum_{n =-1}^{-\infty} a_n .( D - d )^n,
\end{equation}
and regularized value is:
\begin{equation}
 I^{reg} (f):= a_0.
\end{equation}
\subsection {The $D$-dimensional integral}
In this section, we construct the $D$-dimensional integral for Schwartz functions. We will always work with the Euclidean model of space-time. Let $V$ be Euclidean space-time and $\beta$ the positive definite metric on $V$. We denote by $d$ the physical dimension of space-time, and use $D$ for the dimension when viewed as a complex variable. The Lorentz group is replaced by the orthogonal group $O(d)$.\\

Let $W$ be a finite dimensional vector space, we write $S^2 W^*$ for the vector space of symmetric bilinear forms on $w$, $\bar S_+^2 W^*$ and $S_+^2 W^*$ for the subsets of positive semi-definite and positive definite bilinear forms on $W$. Thus
     $$ S_+^2 W^*   \subset  \bar S_+^2 W^*  \subset   S^2 W^*.$$                     
We write $\Cal S (\bar S_+^2 W^*)$ and $\Cal S (S^2 W^*)$ for the spaces of Schwartz functions on $\bar S_+^2 W^*$ and $ S^2 W^*$. By definition, a function on  $\bar S_+^2 W^*$ is a Schwartz function iff it is the restriction of a Schwartz function on $ S^2 W^*$.\\  
Let $E$ be a $n$-dimensional vector space and $F$ a $m$-dimensional subspace of $E$. We view an n-tuple $q:= (q_1, \cdots, q_n)$ with $q_j \in \mathbb R^d$, as an element of the vector space $\mop{Hom}(E, \mathbb R^d)$. The m-tuple $(q_1,  \cdots, q_m)$ is nothing but the restriction $ q_{|F} \in \mop{Hom}(F, \mathbb R^d)$. The bilinear form $q^* (\beta) \in \bar S_+^2 E^*$ is obtained by pulling back $\beta \in (S_+^2 \mathbb R^d)^*$ along $q$. 
\begin{proposition}\cite{RM}
Let $q_1$ and $q_2$ be two elements of $\mop{Hom}(E, \mathbb R^d)$. Then $q_1$ and $q_2$ are in the same $O(d)$-orbit iff $q^*_1 \beta = q^*_2 \beta$.
\end{proposition}
\noindent We suppose that $f$ is already given as a function defined on all of $\bar S_+^2 E^*$. Then we can rewrite our integral as:
\begin{equation}\label{int}
I^d\big|_{k^* \beta} (f) = \int_{\{ q \in \smop{Hom}(E, \mathbb R^d) / q_{|F} = k\}} f(q^* \beta) dq,
\end{equation}
for $k \in \mop{Hom}(F, \mathbb R^d)$. The right hand side depends only on $k^*\beta$, not on $k$ itself. The integral $I^d\big|_C (f)$ is defined only for $\mop{rank}(C) \leq d$.
\subsection {Extrapolation to complex dimensions}
We consider $\phi_B (A)$, defined for all $A$ in $S_+^2 E^*$ and $B$ in $S_+^2 E$ by:
$$\phi_B (A) := \exp (- \tr (AB)).$$
Notice that this is a Schwartz function on $S_+^2 E^*$ by positivity.
\begin{proposition}
Let be $B \in S_+^2 E^*$. We denote by $B_{F^\bot}$, the restriction of $B$ at $F^\bot = (E/F)^* \subseteq E^*$ and by $B^{F^*} \in S_+^2 F^*$ the restriction of  the bilinear form $B$ to $F$. The integral $I^d\big|_C (\phi_B )$ is defined for $C \in \bar S_+^2 F^*$ such that $\mop{rank}(C) \leq d $, and is given by:
\begin{equation}\label{gg}
I^d\big|_C  (\phi_B ) = \pi^{(n-m)d/2} \exp ( - \mop{tr} (C. B^{F^*})).(\det B_{F^\bot})^{-d/2}.
\end{equation}
\end{proposition}
\begin{proof}
We decompose $E$ as follows: $E = F \oplus F^\bot$. Let $(e_1, \cdots, e_m)$ (resp. $(e_{m+1}, \cdots, e_n)$) be an orthogonal basis of $F$ (resp. $F^\bot$) such that $B$ is diagonally on $(e_1, \cdots, e_n)$. We note $\mu_i = B (e_i, e_i)$. The condition $\mop{rank}(C) \leq d$ implies the existence of a symmetric bilinear form $\beta$ on $\mathbb R^d$ such that $C = k^*\beta$ with $k \in \mop{Hom} (F, \mathbb R^d)$. Then we calculate:
\begin{eqnarray*}
I^d\big|_C  (\phi_B ) &=&\int_{\{ q \in \smop{Hom}(E, V) / q_{|F} = k\}} \exp \Big(- \tr (C. B)\Big) dq \\
&=&\int_{\{ q \in \smop{Hom}(E, V) / q_{|F} = k\}} \exp \Big(-\sum_{i=1}^{n}  \mu_i q^* \beta (e_i, e_i)\Big) dq \\
&=& \int_{\{ q \in \smop{Hom}(E, V) / q_{|F} = k\}} \exp \Big(-\sum_{i=1}^{n} \mu_i || q (e_i)||^2 \Big) dq \\  
&=&  \exp \big(-\sum_{i=1}^{m} \mu_i || k_j||^2 \big) \int_{\{ q \in \smop{Hom}(E, V) / q_{|F} = k\}} \exp \Big(-\sum_{i=m+1}^{n} \mu_i || q (e_i)||^2 \Big) dq \\
&=& \exp \Big(- \tr  \big(k^* \beta_{|F}. B^{F^*}\big)\Big) \prod_{i=m+1}^{n} \int \exp \Big( \mu_i || q (e_i)||^2 \Big) dq_i \\
&=& \exp \Big(- \tr \big(C. B^{F^*}\big)\Big) \prod_{i=m+1}^{n} \pi^{d/2} \mu_i^{-d/2}\\
&=& \pi^{(n-m)d/2}  \exp \Big(- \tr \big(C. B^{F^*}\big)\Big) (\det B_{F^\bot})^{-d/2}.
\end{eqnarray*}
\end{proof}
We write $E$ as $E \cong F \oplus G$ and describe $B$ by a block matrix $(B_{ij})$ that respects the decomposition $E^* = F^* \oplus G^*$. We obtain:
\begin{equation}\label{hh5}
 B_{F^\bot} = B_{22}, \hskip 3cm  B^{F^*} = B_{11} - B_{12} B_{22}^{-1} B_{21}.
\end{equation}
If the decomposition of $E$ diagonalizes the matrix $B$, then we get $B^{F^*} = B_{11}$. 
\begin{theorem}\cite{RM} \label{thmdis}
There is unique family of distributions $I^D \big|_C \in \Cal S (\bar S_+^2 E^*)'$ that satisfies
\begin{equation}\label{th5}
 I^D \big|_C (\phi_B ) = \pi^{(n-m)D/2} \exp ( - \mop{tr} (C. B^{F^*})).(\det B_{F^\bot})^{-D/2}
\end{equation} 
for all $B \in  S_+^2 E$, $ C \in \bar S_+^2 F^*$, and $D \in \mathbb C$. In addition, these distributions piece together to a continuous linear map
$$ I = (I^D)_{D \in \mathbb C} : \Cal S ( \bar S_+^2 E^*) \longrightarrow \Cal O (\mathbb C , \Cal S ( \bar S_+^2 F^*) ).$$
\end{theorem}
\begin{definition}
The operator $I^D$ is called $D$-dimensional integral with parameters. If $F = \{0\}$, it is called $D$-dimensional integral.     
\end{definition}
\subsection {Construction of the $D$-dimensional integral with parameters}
To prove Theorem \ref{thmdis}, we present two descriptions for a distribution $I^D \big|_C$ satisfying \eqref{th5}. The first one only works for $\Re D > n-1$. The second one works for $\Re < 0$ and is used to extend the first description to all $D \in \mathbb C$.\\
For all $l \in \mathbb N$, $x \in \mathbb C$, we define:
$$\Gamma_l (x) := \pi^{l(l-1)/4} \prod_{j=0}^{l-1} \Gamma(x-\frac{j}{2}).$$
For $\Re D > n-1$ and $C \in S_+^2 F^*$, we define:
$$ A \longmapsto \rho^D ( A, C) :=  \pi^{(n-m)D/2} . \frac{\Gamma_m(D/2)}{\Gamma_n(D/2)}.\frac{(\det A)^{(D-n-1)/2}}{(\det C)^{(D-m-1)/2}} . \delta (A_F - C).$$
This is a well-defined distribution because the function $(\det A)^{-1}$ is locally integrable on $S^2 E^*$.
\begin{lemma} \cite{RM} \label{lem1}
Let $\Re D > n-1$, $C \in S_+^2 F^*$ and $B \in S_+^2 E$. Then:
$$\int_{S_+^2 E^*} \rho^D ( A, C) \phi_B (A) dA = \pi^{(n-m)D/2} \exp ( - \mop{tr} (C. B^{F^*})) . (\det B_{F^\bot})^{-D/2}.$$  
\end{lemma}
\begin{proof}
Both sides of the equation are defined independently of the choice of basis, they only use the subspace $F \subseteq E$, the volume forms on $E$ and $F$, and the positive definite bilinear forms $C$ and $B$. Let $(x_i)_{1\leq i\leq n}$ be a basis of $E$ such that $F = < x_1, \ldots, x_m >$ and $G = < x_{m+1}, \ldots, x_n >$. Let $(x^*_i)_{1\leq i\leq n}$ be the corresponding dual basis for $E^*$. We put:
 $$ G := F^{\bot_B} =  < x_{m+1}, \ldots, x_n > \subseteq E, \hskip 2cm G^* = F^{\bot} \subseteq E^*,$$
The positive definite bilinear forms $B$ and $C$ are then written in the form:
$$ C = \begin{pmatrix}
    c_1 & \cdots & 0 \\ 
    \vdots & \ddots & \vdots \\ 
    0      & \cdots & c_n 
\end{pmatrix}
,\hskip3cm B  =
\left( \begin{array}{cc}
B^{F^*} & 0 \\
0 & B_{F^\bot} 
\end{array}\right).
$$
We represent an element $A$ of $S_+^2 E^*$ as a block matrix $(A_{ij})$ with respect to the decomposition $E = F \oplus G$. If $A$ is positive definite, so is  $A_{11}$. Hence $A_{11}$ is invertible and we can define:
$$Y := A_{21} A_{11}^{-1} , \hskip 2cm  Y^* :=  A_{11}^{-1} A_{12}, \hskip 2cm  X :=  A_{22} - A_{21} A_{11}^{-1} A_{12}.$$
The computation
$$
\left( \begin{array}{cc}
1 & 0 \\
Y & 1 
\end{array}\right) . \left( \begin{array}{cc}
A_{11} & 0 \\
0 & X 
\end{array}\right) . \left( \begin{array}{cc}
1 & Y^* \\
0 & 1 
\end{array}\right) = \left( \begin{array}{cc}
A_{11} & A_{12} \\
A_{21} & A_{22}
\end{array}\right) $$  
shows that  $A > 0$ iff $A_{11} > 0$ and $X > 0$ and that $\det A = \det A_{11} . \det X$. We apply the change of variables: $A \longmapsto (A_{11}, X ,A_{21})$, which identifies
$$ S_+^2 E^* \cong S_+^2 F^* \times S_+^2 F^{\bot} \times \mop{Hom} (F^*, F^{\bot}).$$
Its Jacobian has determinant $1$ everywhere. Simplifying first the $\delta$-function and then the Gaussian integral for $A_{21}$, we obtain:
\begin{eqnarray*}
I &=& \int_{S_+^2 E^*} \rho^D ( A, C) \phi_B (A) dA\\
&=&\int_{S_+^2 E^*}  \pi^{(n-m)D/2}  \frac{\Gamma_m(D/2)}{\Gamma_n(D/2)} \frac{(\det A)^{(D-n-1)/2}}{(\det C)^{(D-m-1)/2}} \exp \big(- \tr (A B)\big). \delta (A_F - C) dA \\
&=& \int_{{S_+^2 F^*} \times {S_+^2 F^{\bot}} \times{\smop{Hom} (F^*, F^{\bot})}} \pi^{(n-m)D/2}  \frac{\Gamma_m(D/2)}{\Gamma_n(D/2)} \frac{(\det A_{11} . \det X)^{(D-n-1)/2}}{(\det C)^{(D-m-1)/2}}\\
&&\\
&& \hskip 2.5cm \exp \big(- \tr (A_{11}B_{11}) - \tr (A_{22}B_{22}) \big). \delta (A_F - C) d A_{11} d A_{21} d X \\
&&\\
&=& \int_{{S_+^2 F^*} \times {S_+^2 F^{\bot}} \times{\smop{Hom} (F^*, F^{\bot})}} \pi^{(n-m)D/2}  \frac{\Gamma_m(D/2)}{\Gamma_n(D/2)} \frac{(\det A_{11} . \det X)^{(D-n-1)/2}}{(\det C)^{(D-m-1)/2}}\\
&&\\
&& \hskip 2.5cm \exp \big(- \tr (A_{11}B_{11}) - \tr (X + A_{21} A_{11}^{-1} A_{12}) B_{22} \big). \delta (A_F - C) d A_{11} d A_{21} d X \\
&&\\
&=& \pi^{(n-m)D/2} . \frac{\Gamma_m(D/2)}{\Gamma_n(D/2)} (\det C)^{(m-n)/2}  \exp \big( - \tr (B^{F^*} C)\big)\\
&&\\
&& \hskip 2.5cm \int_{S_+^2 F^{\bot}} (\det X)^{(D-n-1)/2}  \exp \big( -tr (B_F^{\bot} X)\big)\\
&& \hskip 2.5cm \int_{\smop{Hom} (F^*, F^{\bot})} \exp \big( - \tr (B_{F^{\bot}} A_{21} C^{-1} A_{21}^t)\big) d A_{21} d X.
\end{eqnarray*} 
Let $J = \int_{\smop{Hom} (F^*, F^{\bot})} \exp \big( - \tr (B_{F^{\bot}} A_{21} C^{-1} A_{21}^t)\big) d A_{21}$, we perform the following change of variables:
$$A'_{21} = C^{-1/2} A_{21},$$
which implies that:
$$dA_{21} = (\det C)^{m(n-m)/2m} dA'_{21} = (\det C)^{(n-m)/2} dA'_{21},$$
and $J$ becomes:
\begin{eqnarray*}
J&=& (\det C)^{(n-m)/2} \int_{\smop{Hom} (F^*, F^{\bot})} \exp \big( - \tr (B_{F^{\bot}} A_{21}  A_{21}^t)\big) d A_{21}.
\end{eqnarray*}
We perform another change of variables: 
$$A'_{21} = (B_{F^\bot})^{1/2} A_{21} = (b_{F^\bot})^{1/2} A_{21},$$
which implies that:
$$dA_{21} = (\det B_{F^{\bot}})^{-m(n-m)/2(n-m)} dA'_{21} = (\det B_{F^{\bot}})^{-m/2} dA'_{21},$$
and we obtain:
\begin{eqnarray*}
J&=& (\det C)^{(n-m)/2} (\det B_{F^{\bot}})^{-m/2}  \int_{\smop{Hom} (F^*, F^{\bot})} \exp \big( - \tr ( A_{21} A_{21}^t)\big) d A_{21}\\
&=& \pi^{m(n-m)/2} (\det C)^{(n-m)/2} (\det B_{F^{\bot}})^{-m/2} ,
\end{eqnarray*}
and the integral $I$ becomes:
\begin{eqnarray*}
I&=& \pi^{(n-m)D/2} . \frac{\Gamma_m(D/2)}{\Gamma_n(D/2)} (2\pi)^{m(n-m)/2}  \exp \big( - \tr (B^{F^*} C)\big)\\
&&\hskip3cm \int_{S_+^2 F^{\bot}} (\det B_{F^{\bot}})^{-m/2} (\det X)^{(D-n-1)/2}  \exp \big( - \tr (B_F^{\bot} X)\big) d X.
\end{eqnarray*}
Soit $K = \int_{S_+^2 F^{\bot}} (\det B_{F^{\bot}})^{-m/2} (\det X)^{(D-n-1)/2}  \exp \big( - \tr (B_{F^{\bot}} X)\big) d X$.
We perform the following change of variables: 
$$X' = B_{F^{\bot}} X = b_{F^{\bot}} X,$$
which implies that:
$$dX = (\det B_{F^{\bot}})^{-(n-m)(n-m+1)/2(n-m)} dX' = (\det B_{F^{\bot}})^{-(n-m+1)/2} dX',$$
and the integral $K$ becomes:
\begin{eqnarray*}
K&=& (\det B_{F^{\bot}})^{- m/2} (\det B_{F^{\bot}})^{-(n-m+1)/2} \int_{S_+^2 F^{\bot}}  (\det {b_{F^{\bot}}}^{-1} X)^{(D-n-1)/2}  \exp \big( - \tr X\big) d X\\
&=& (\det B_{F^{\bot}})^{[- m - (n-m+1) - (D-n-1)]/2 } \int_{S_+^2 F^{\bot}}  (\det  X)^{(D-n-1)/2}  \exp \big( - \tr X \big) d X\\
&=& (\det B_{F^{\bot}})^{- D/2 } \int_{S_+^2 F^{\bot}}  (\det  X)^{(D-n-1)/2}  \exp \big( - \tr X \big) d X.
\end{eqnarray*}
Now we decompose $G = F^{\bot} := G_1 \oplus G_2$ with $\dim G_1 = 1$. we write $ X = (X_{ij})$, et $X_{11} = x_{11}$, and we put: 
$$ T :=  X_{21} X_{11}^{-1} , \hskip 2cm  T^* :=  X_{11}^{-1} X_{12}, \hskip 2cm  L :=  X_{22} - X_{21} X_{11}^{-1} X_{12}.$$
By calculation we obtain
$$
\left( \begin{array}{cc}
1 & 0 \\
T & 1 
\end{array}\right) . \left( \begin{array}{cc}
X_{11} & 0 \\
0 & L 
\end{array}\right) . \left( \begin{array}{cc}
1 & T^* \\
0 & 1 
\end{array}\right) = \left( \begin{array}{cc}
X_{11} & X_{12} \\
X_{21} & X_{22}
\end{array}\right).$$  
This equality shows that $X > 0$ iff $X_{11} > 0$ and $L > 0$ and that $\det X = \det X_{11} . \det L$. We use the following change of variables: $X \longmapsto (X_{11}, L ,X_{21})$. The integral $K$ is written:  
\begin{eqnarray*}
K&=&(\det B_{F^{\bot}})^{- D/2 } \int_{{S_+^2 G_1} \times {S_+^2 G_2} \times{\smop{Hom} (G_1, G_2)}} (\det  X_{11})^{(D-n-1)/2}  \exp \big( - \tr X_{11} \big) \\
&&\\
&&\hskip3cm (\det  L)^{(D-n-1)/2} \exp \big( - \tr (L + X_{21} X_{11}^{-1} X_{12}) \big) d X_{11} d L  d X_{21}  \\
&&\\
&=&(\det B_{F^{\bot}})^{- D/2 } \int_{{S_+^2 G_1} \times {S_+^2 G_2}} (\det  X_{11})^{(D-n-1)/2}  \exp \big( - \tr X_{11} \big) (\det  L)^{(D-n-1)/2} \\
&&\\
&& \hskip3cm \exp \big( - \tr (L ) \big) \underbrace{\int_{\smop{Hom} (G_1, G_2)}\exp \big( - \tr ( X_{21} X_{11}^{-1} X_{12}) \big) d X_{21}}_{\pi^{(n-m-1)/2} (\det  X_{11})^{(n-m-1)/2}} \; d X_{11} d L.\\   
&=& \pi^{(n-m-1)/2} (\det B_{F^{\bot}})^{- D/2 } \int_{S_+^2 G_1}  (\det  X_{11})^{[(D-n-1)/2 + (n-m-1)/2]}  \exp (- X_{11})  d X_{11}\\
&&\\
&&\hskip3cm\int_{S_+^2 G_2} (\det L)^{(D-n-1)/2} \exp \big( -\tr (L)\big)d L   \\
&&\\
&=& \pi^{(n-m-1)/2} (\det B_{F^{\bot}})^{- D/2 } \int_{S_+^2 G_1}  X_{11}^{(D-m-2)/2}  \exp (- X_{11})  d X_{11}\\
&&\\
&&\hskip3cm\int_{S_+^2 G_2} (\det L)^{(D-n-1)/2} \exp \big( -\tr (L)\big)d L   \\
&&\\
&=& \pi^{(n-m-1)/2} (\det B_{F^{\bot}})^{- D/2} \Gamma \big(D/2 - m/2\big). \int_{S_+^2 G_2} (\det L)^{(D-n-1)/2} \exp \big( - \tr (L)\big) dL .
\end{eqnarray*}
we repeat this process $n - m + 1$ times, which yieds a product of Gamma functions and powers of $\pi$:
\begin{eqnarray*}
I&=& \pi^{(n-m)D/2} . \frac{\Gamma_m(D/2)}{\Gamma_n(D/2)} \exp \big( -\tr (B^{F^*} C)\big) (\det B_{F^{\bot}})^{-D/2} .\pi^{m(n-m)/2}. \pi^{(n-m-1)/2}\\
&&\\
&& .\Gamma \big(D/2 - m/2\big). \pi^{(n-m-2)/2}.\Gamma \big(D/2 - (m+1)/2\big) \cdots \pi^{0}.\Gamma \big(D/2 - (n-1)/2\big)\\
&&\\
&=& \pi^{(n-m)D/2} \frac{\pi^{\frac{(m-n)(m+n-1)}{4}}}{\prod_{j = m}^{n-1} \Gamma \big(D/2 - j/2\big)} \exp \big( -\tr (B^{F^*} C)\big) (\det B_{F^{\bot}})^{-D/2} \\
&& \pi^{(n-m)m/2} \pi^{\sum_{j = m+1}^{n} (n-j)/2} . \prod_{j = m}^{n-1} \Gamma \big(D/2 - j/2\big)\\
&=& \pi^{(n-m)D/2}  \exp \big( -\tr (B^{F^*} C)\big) (\det B_{F^{\bot}})^{-D/2}. 
\end{eqnarray*}
Hence the final result:
$$\int_{S_+^2 E^*} \rho^D ( A, C) \phi_B (A) dA = \pi^{(n-m)D/2} \exp ( -\tr (C. B^{F^*})) . (\det B_{F^\bot})^{-D/2}.$$ 
\end{proof}
Thus, we can define $I^D \big|_C $ for ${\Re} D > n-1$ and $C \in S_+^2 F^*$ by: 
$$ I^D \big|_C : = \rho^D ( . , C).$$
\begin{definition}
We define $I^D_{E, F} : \Cal S ( \bar S_+^2 E^*) \longrightarrow  \Cal S ( \bar S_+^2 F^*)$ for all $C \in S_+^2 F^*$ and $f \in \Cal S ( \bar S_+^2 E^*)$, by:
\begin{equation}
I^D_{E, F} (f) (C) := I^D \big|_C (f).
\end{equation}
\end{definition}
\begin{proposition}\label{prosc} Let $E$, $F$ and $G$ be three vector spaces such that $G \subseteq F \subseteq E$. The following diagram is commutative:
\diagramme{\xymatrix{
\Cal S ( \bar S_+^2 E^*) \ar[rrd]_{I^D_{E, G}}\ar[rr]^{I^D_{E, F}} &&\Cal S ( \bar S_+^2 F^*) \ar[d]^{I^D_{F, G}} \\
&& \Cal S ( \bar S_+^2 G^*) 
 }}
\end{proposition}
\begin{proof}
Let $C \in S_+^2 F^*$, $C' \in S_+^2 G^*$ with $C_{|G} = C'$ and $B \in S_+^2 E$, we have:
\begin{eqnarray*}
I^D_{F, G} \big(I^D_{E, F} (\phi_B)\big)(C') &=& \int_{S_+^2 F^*} \rho^D ( C, C') \int_{S_+^2 E^*} \rho^D ( A, C) \phi_B (A) dA dC\\
&=& \pi^{(n-p)D/2} \frac{\Gamma_p(D/2)}{\Gamma_m(D/2)} \frac{\Gamma_m(D/2)}{\Gamma_n(D/2)} \int_{S_+^2 F^*}  \frac{(\det C)^{(D-m-1)/2}}{(\det C')^{(D-p-1)/2}}  \\
&& \int_{S_+^2 E^*} \frac{(\det A)^{(D-n-1)/2}}{(\det C)^{(D-m-1)/2}}  \phi_B (A)  \delta (C_G - C') \delta (A_F - C) dA dC \\
&=& \pi^{(n-p)D/2} \frac{\Gamma_p(D/2)}{\Gamma_n(D/2)} \int_{S_+^2 F^*} \frac{(\det A)^{(D-n-1)/2}}{(\det C')^{(D-p-1)/2}}  \\
&& \int_{S_+^2 E^*} \phi_B (A)  \delta (C_G - C') \delta (A_F - C) dA dC \\
&=&\hskip-0.3cm\int_{S_+^2 E^*} \pi^{(n-p)D/2} \frac{\Gamma_p(D/2)}{\Gamma_n(D/2)} \frac{(\det A)^{(D-n-1)/2}}{(\det C')^{(D-p-1)/2}} \phi_B (A) \delta (A_G - C') dA\\
&=&\int_{S_+^2 E^*} \rho^D ( A, C') \phi_B (A) dA \\
&=&I^D_{E, G} \big(\phi_B \big) (C').
\end{eqnarray*}
Hence: $I^D_{E, G} = I^D_{F, G} \circ I^D_{E, F}$. 
\end{proof}
\subsection {$D$-dimensional integral of a Feynman type function} 
\begin{definition}\label{fonc fey}
A function $f \in {S_+^2 E^*}$ is of Feynman type if it is written in the form:
\begin{equation}
f(A) = \frac{P(A)}{\prod_{j=1}^l \big( \tr (AB_j) + m_j^2\big)},
\end{equation}
where $P$ is a polynomial, $B_j \in {\bar S_+^2 E^*}$ such that $B(t):= \big(\sum t_j B_j \big) \in {S_+^2 E^*}$ for all $t_j >0$, and the $m_j$ are positive real numbers. In particular, for all subspace $F\subset E$, $(B(t))_{F^\bot}$ is non-degenerate. The function $f$ is without poles in ${\bar S_+^2 E^*}$. 
\end{definition}
\begin{remark}
We note here that the functions of Feynman type form an algebra. For the product that is immediate, and for the sum it is sufficient to reduce to the same denominator. 
\end{remark}
\begin{proposition}\cite{PE}
If $f$ is a Feynman type function then $I_{E, F}^D (f)$ extends to the whole complex plane to a meromorphic function of the variable $D$ for all subspace $F\subset E$.
\end{proposition}
\begin{proof}
We first prove the proposition for $P = 1$, the case where $P \neq 1$ will be proved later. Firstly we use the following equality:
$$ \int_0^\infty \exp (-a t) dt = a^{-1} \;\; \text{for all} \;\;  a > 0.$$
Let $t = (t_1, \cdots, t_l)$ and $B(t) = \sum_{j=1}^l t_j B_j$. Since we have:
\begin{eqnarray*}
f(A) &=& \frac{1}{\prod_{j=1}^l \big( \tr (AB_j) + m_j^2\big)}\\
&=& \prod_{j=1}^l   \big( \tr (AB_j) + m_j^2\big)^{-1} \\
&=& \prod_{j=1}^l \int_{t_j >0}  \exp \big( - t_j (\tr (AB_j) +m_j^2) \big) dt_j\\
&=&  \int_{t_j >0}  \exp \big( - \sum_{j=1}^l \tr A (t_jB_j)\big) \exp \big( - \sum_{j=1}^l t_j m_j^2 \big)  dt_1 \cdots dt_l\\
&=& \int_{t_j >0}  \exp \big( - \tr A (\sum_{j=1}^l t_jB_j)\big) \exp \big( - \sum_{j=1}^l t_j m_j^2 \big) dt,
\end{eqnarray*}
we obtain for all $C \in {S_+^2 F^*}$:
\begin{eqnarray*}
I_{E, F}^D (f) (C) &=& \int_{t_j >0 }  \exp \big( - \sum_{j=1}^l t_j m_j^2 \big)  I_{E, F}^D \Big(\exp \big( - \tr (C \sum_{j=1}^l t_jB_j)\big)\Big) dt.
\end{eqnarray*}
Using Lemma \ref{lem1}, we can write:
\begin{equation}
 I_{E, F}^D (f) (C) = \pi^{(n-m)D/2} \int_{t_j >0 }  \exp \big( - \sum_{j=1}^l t_j m_j^2 -\tr (C. B(t)^{F^*}) \big) . (\det B(t)_{F^\bot})^{-D/2} dt.
\end{equation}
To finish the proof of the proposition we use Bernstein's theorem \cite{Ber} (see also \cite{Cou}) and a corollary.
\begin{theorem} \textbf{(Bernstein's theorem)}
Let $Q$ be a polynomial with $l$ variables. Then there exists a differential operator $L(D)$ in $l$ variables, with coefficients which depend on $D$ and a polynomial $q$ in $D$ such that:
\begin{equation}\label{bern}
L(D) Q^{- D/2} = q(D) Q^{-1 - D/2}.
\end{equation}
\end{theorem}
\begin{corollary}\label{coroll}
We assume that $ Q $ takes positive values for $t_j > 0$. Let $g$ be a function with rapid decay defined for $t_j > 0$, and such that its derivatives are also rapidly decreasing. Then the integral
\begin{equation}\label{berncorbern}
I(D, g) := \int_{t_j >0 } g(t) Q^{- D/2} (t) dt,
\end{equation}
converges for $\Re << 0$, and extends to a meromorphic function in the whole complex plane.
\end{corollary}
\begin{proof}[Preuve]
Using the formulas \eqref{bern} et \eqref{berncorbern} we can write:
\begin{eqnarray*} 
I(D + 2, g) &=& \int_{t_j >0 } g(t) Q^{- D/2 -1} (t) dt\\
&=& \int_{t_j >0 } g(t) q(D)^{-1} L(D) Q^{- D/2} (t) dt\\
&=&  q(D)^{-1} \int_{t_j >0 } g(t) L(D) Q^{- D/2} (t) dt.
\end{eqnarray*}
Using the following induction:
$$g(t) \longrightarrow L(D)^* (g(t))  \;\;\;\;\text{et} \;\;\;\; L(D) Q^{- D/2} (t)\longrightarrow Q^{- D/2} (t)$$
we obtain:
\begin{eqnarray*} 
I(D + 2, g) &=& q(D)^{-1} \int_{t_j >0 } L(D)^* (g(t)) Q^{- D/2} (t) dt + C(D)\\
&=& q(D)^{-1} I(D , L(D)^* (g(t))) + C(D),
\end{eqnarray*}
where $C(D)$ is a sum of integrals of the same type on the boundary. This term is meromorphic by the induction hypothesis, since the boundary of $(\mathbb R_+)^l$ can be written as the union of $l$ copies of $(\mathbb R_+)^{l-1}$ and strata in dimension $\leq l-2$.
\end{proof}
Using Bernstein's theorem and the previous corollary, we can conclude that:
$$I_{E, F}^D (f) (C) = \pi^{(n-m)D/2} \int_{t_j >0 }  \exp \big( - \sum_{j=1}^l t_j m_j^2 -\tr (C. B(t)^{F^*}) \big) . (\det B(t)_{F^\bot})^{-D/2} dt$$
extends to a meromorphic function into $D$, which proves the proposition for $P = 1$.\\
For $P \neq 1$, we obtain the same results: considering
 $$f(A) = \frac{P(A)}{\prod_{j=1}^l \big( \tr (AB_j) + m_j^2\big)},$$
the integral $I_{E, F}^D (f) (C)$ will be written in the form:
\begin{eqnarray*}
I_{E, F}^D (f) (C) &=& \int_{t_j >0 }  \exp \big( - \sum_{j=1}^l t_j m_j^2 \big)  I_{E, F}^D \big( P \exp ( - \tr (. B(t))\big) (C)  dt\\
&=& \int_{t_j >0 }  \exp \big( - \sum_{j=1}^l t_j m_j^2 \big)  I_{E, F}^D \big( P \phi_{B(t)} \big) (C)  dt\\
&=&\int_{t_j >0 }  \exp \big( - \sum_{j=1}^l t_j m_j^2 \big)  I_{E, F}^D \big( \bar P (\partial_B) \phi_{(B(t))} \big) (C)  dt,
\end{eqnarray*}
where $\bar P (\partial_B)$ is a constant coefficient differential operator, defined as follows: for $A = (a_{ij})$ et $B(t) = ( b_{ij}(t))$ we have:
$$\tr (A B(t)) = \sum_{k,i} a_{ik} b_{ki} (t),$$
and then:
\begin{eqnarray*}
- \frac{\partial}{\partial b_{qp}} \exp \big( -\tr (A. B(t)) \big) &=& a_{pq} \exp \big( - \sum_{k,i} a_{ik} b_{ki} (t) \big) \\
&=& a_{pq} \exp \big( -\tr (A. B(t)) \big), 
\end{eqnarray*}
what defines the polynomial $\bar P$ when $P(A) = a_{pq}$, and therefore for all $P$ by iterating the process. As a result:
\begin{eqnarray*}
I_{E, F}^D \big( \bar P (\partial_B) \phi_{(B(t))} \big) (C) &=& \bar P (\partial_B) I_{E, F}^D \big(\phi_{(B(t))} \big) (C)\\
&=& \bar P (\partial_B) \pi^{(n-m)D/2} \exp \big( -\tr (C. B(t)^{F^*}) \big) . (\det B(t)_{F^\bot})^{-D/2}\\
&=&\sum_{r=0}^{\smop {deg} \bar P} g_{r, C} (t) (\det B(t)_{F^\bot})^{-D/2 - r},   
\end{eqnarray*}
where $t \longmapsto g_{r, C} (t)$ is a Schwartz function. Then we obtain: 
\begin{eqnarray*}
I_{E, F}^D (f) (C) &=& \int_{t_j >0 }  \exp \big(- \sum_{j=1}^l t_j m_j^2 \big)  \sum_{r=0}^{\smop {deg} \bar P} g_{r, C} (t) (\det B(t)_{F^\bot})^{-D/2 - r} dt.
\end{eqnarray*}
Hence, under the same conditions as for $P = 1$, the integral $I_{E, F}^D (f)$, if $P\neq 1$, extends to the whole complex plane to a meromorphic function for the complex variable $D$. 
\end{proof}
We denote by $\Cal F (\bar S_+^2 E^*)$ the space of Feynman type functions on $\bar S_+^2 E^*$ and by $\wt {\Cal F} (\bar S_+^2 E^*)$ the space of functions on $\mathbb C \times \bar S_+^2 E^*$, meromorphic in the first variable, equal to $I_{E', E}^D (g)$ for some function $g \in \Cal F (\bar S_+^2{E'}^*)$, where $E'$  is a vector space containing $E$. We use this representation:  
\begin{eqnarray*}
g (A') &=& \int_{t_j >0 }  \bar P (\partial_B) \exp \big( -\tr (A. B(t)) \big) \exp \big( - \sum_{j=1}^l t_j m_j^2 \big) dt,
\end{eqnarray*}
for all $A' \in S_+^2 E'^*$, let also:
\begin{eqnarray*}
g &=& \int_{t_j >0 }  \bar P (\partial) \phi_{(B(t))} \exp \big( - \sum_{j=1}^l t_j m_j^2 \big) dt.
\end{eqnarray*}
We have then for all $A \in S_+^2 E^*$ :
$$ f(A) = I_{E', E}^D (g)(A) = \int_{t_j >0 } \exp \big( - \sum_{j=1}^l t_j m_j^2 \big) \underbrace{I_{E', E}^D \big(\bar P (\partial_B) \exp ( -\tr (.. B(t)))\big)(A)}_{\varphi_t(A)} dt.$$
Clearly $\varphi_t \in \Cal S (\bar S_+^2 E^*)$. Then we put:
$$ I_{E, F}^D (f) :=  \int_{t_j >0 } \exp \big( - \sum_{j=1}^l t_j m_j^2 \big) I_{E, F}^D (\varphi_t) dt.$$
By means of the Proposition \ref{prosc}, we can write for all $C \in \bar S_+^2 F^*$ :
\begin{eqnarray*}
I_{E, F}^D (f)(C) &=& \int_{t_j >0 } \exp \big( - \sum_{j=1}^l t_j m_j^2 \big) I_{E, F}^D I_{E', E}^D \big(\bar P (\partial) \phi_{(B(t))}\big) (C) dt\\
&=& \int_{t_j >0 } \exp \big( - \sum_{j=1}^l t_j m_j^2 \big) I_{E', F}^D \big(\bar P (\partial) \phi_{(B(t))}\big) (C) dt\\
&=& I_{E', F}^D (g)(C). 
\end{eqnarray*}
\begin{corollary}\label{corool}
We have: $$ I_{E, F}^D : \wt {\Cal F} (\bar S_+^2 E^*)\longrightarrow \wt {\Cal F} (\bar S_+^2 F^*),$$
and if $G \subset F \subset E$, the following diagram is commutative:
$$
\xymatrix{
\wt {\Cal F} ( \bar S_+^2 E^*) \ar[rrd]_{ I^D_{E, G}}\ar[rr]^{ I^D_{E, F}} &&\wt {\Cal F} ({\bar S}_+^2 F^*) \ar[d]^{ I^D_{F, G}} \\
&& \wt {\Cal F} ({\bar S}_+^2 G^*) 
 }
$$
in other words:
$$I_{E, G}^D = I_{F, G}^D \circ I_{E, F}^D.$$ 
\end{corollary}
\section {Renormalization of the Feynman integral} 
\subsection {The target algebra $\mop{End} \wt {\Cal B}$ and convolution product} 
Let $\Cal T$ be a quantum field theory and let $\Gamma$ be a connected Feynman graph of $\Cal T$. We recall that $W_\Gamma$ is the momenta space of graph $\Gamma$, and $E = \Cal E (\Gamma)$. In this section, we define the target algebra $\wt {\Cal B}$, the product $\bullet$ and the convolution product $\divideontimes$ analogously to Section \S 7. We put:
\begin{equation}
 \wt V_\Gamma :=  \wt{\Cal F} (\bar S^2_+ E^*).
\end{equation} 
For $\Gamma = \Gamma_1 \cdots \Gamma_r$ connected, we set:
$$\wt V_\Gamma = \bigotimes_{j\in \{1, \cdots, r\}} \wt V_{\Gamma_j},$$
\begin{equation}\label{cabi}
\wt {\Cal B} := \prod_{\Gamma \in \Cal T} \wt V_{\Gamma}.
\end{equation}
The product $\bullet$ is defined on the elements of $\mop{End} \wt {\Cal B}$ similarly to the Paragraph \S \ref{sch}. It is compatible with the composition product. In other words, for the linear maps:
$$\wt a_1 : V_{\Gamma_1} \longrightarrow V_{\Gamma'_1} ; \;\;\;\;\; \wt b_1 : V_{\Gamma'_1} \longrightarrow V_{\Gamma''_1} ,$$
$$\wt a_2 : V_{\Gamma_2} \longrightarrow V_{\Gamma'_2} ; \;\;\;\;\; \wt b_2 : V_{\Gamma'_2} \longrightarrow V_{\Gamma''_2} ,$$
we obtain the following result:
\begin{equation} \label{ctu}
(\wt b_1 \circ \wt a_1) \bullet (\wt b_2 \circ \wt a_2) = (\wt b_1 \bullet \wt b_2) \circ (\wt a_1 \bullet \wt a_2).
\end{equation}
We denote by $\diamond$ the opposite of composition product in $\mop {End} \wt {\Cal B}$. Then we define a convolution product $\divideontimes$ for all $\varphi$, $\psi \in {\Cal L} ({\Cal D}_{\Cal T}, \mop {End} \wt {\Cal B})$ by: 
\begin{equation}
 \varphi \divideontimes \psi :=  \diamond  ( \varphi \otimes \psi )  \Delta.
\end{equation}
In other words, for all specified graphs $\bar\gamma$, $\bar\Gamma$ such that $\bar\gamma \subset \bar\Gamma$ we have:
\begin{equation}
(\varphi \divideontimes \psi) (\bar\Gamma, \bar\gamma ) = \sum_{\substack{\bar\delta \subseteq \bar\gamma \\ \bar\gamma / \bar\delta \in \Cal T }}\psi( \bar\Gamma / \bar\delta, \bar\gamma / \bar\delta)\circ \varphi( \bar \Gamma, \bar\delta ).
\end{equation}
Similarly to the Paragraph \S 7.4 we obtain the following theorem:
\begin{theorem}
Let $G$ be the set of morphisms of unitary algebras: ${\Cal D}_{\Cal T} \longrightarrow (\mop {End} \wt {\Cal B}, \bullet)$. Equipped with the product $\divideontimes$, the set $G$ is a group. 
\end{theorem}
\subsection {Feynman integral} 
Let $\Gamma$, $\gamma$ and $\delta$ be three Feynman graphs such that $\delta \subseteq \gamma \subseteq \Gamma$. We put:
$$ E := \Cal E (\Gamma),    \;\;\;\;\;\;\;\; F := \Cal E (\Gamma/\delta) \;\;\;\;\;\text{and} \;\;\;\;\; G:= \Cal E (\Gamma/\gamma) = \Cal E (\Gamma/\delta \Big/ \gamma/\delta),$$
we have $ G \subseteq F \subseteq E$. We define the Feynman integral $\wt I^D_{\Gamma, \gamma}$ by:
\begin{eqnarray*}
\wt I^D_{\Gamma, \gamma} := I^D_{E, F} : \Cal S (\bar S_+^2 E^*) &\longrightarrow& \Cal S (\bar S_+^2 F^*) 
\end{eqnarray*} 
This expression is holomorphic in $D$, in other words it defines an operator:
\begin{equation}
 \wt I_{\Gamma, \gamma} = (I^D)_{D \in \mathbb C} : \Cal S ( \bar S_+^2 E^*) \longrightarrow \Cal O (\mathbb C , \Cal S ( \bar S_+^2 F^*)).
\end{equation}
\begin{theorem}
Let $\Gamma$, $\gamma$ and $\delta$ be three Feynman graphs such that $\delta \subseteq \gamma \subseteq \Gamma$. We have:
\begin{equation}
\wt I^D_{\Gamma, \gamma} = \wt I^D_{\Gamma, \delta} \circ \wt I^D_{\Gamma/\delta, \gamma/\delta}.
\end{equation}
\end{theorem}
\begin{proof}
This is a direct corollary of Proposition \ref{prosc}.
\end{proof}
We adopt the notation $\wt I^D_{\Gamma, \gamma}$ for $I^D_{E, F} : \wt {\Cal F} (\bar S_+^2 E^*) \longrightarrow \wt {\Cal F} (\bar S_+^2 F^*)$. It extends to a meromorphic function of the complex variable $D$.
\begin{theorem}
Let $\Gamma$, $\gamma$ and $\delta$ be three Feynman graphs such that $\delta \subseteq \gamma \subseteq \Gamma$. Then:
\begin{equation}
\wt I^D_{\Gamma, \gamma} = \wt I^D_{\Gamma/\delta, \gamma/\delta} \circ \wt I^D_{\Gamma, \delta}.
\end{equation}
\end{theorem}
\begin{proof}
This is a direct consequence of corollary \ref{corool}.
\end{proof}
The Feynman rules are defined for $U = \Cal E (\mop{res} \Gamma)$ by:
$$\wt I_{\Gamma, \Gamma} \big(\varphi (\Gamma)\big) \in \wt {\Cal F} (\bar S_+^2 U^*),$$
where $\varphi (\Gamma)$ is the integrand defined by the formula \eqref{propa}, which can also be written as being a Feynman type function on $\bar S_+^2 U^*$, in the form:
\begin{equation}
\varphi(\Gamma)(p) = \prod_{v\in \Cal V (\Gamma)} g_v  \prod_{\{e,\sigma(e)\}, \sigma(e)\neq e} G_{e \sigma(e)} (p^*\beta(e , e))\prod_{\sigma(e)= e} G_e (p^*\beta(e , e)).
\end{equation}
\subsection {Birkhoff decomposition} 
Let $\varphi$ be an element of the group $G (k[z^{-1} , z]])$, that is to say a character of ${\Cal D}_{\Cal T}$ with values in the unitary commutative algebra  ${\Cal A} : = \mop {End} \wt {\Cal B} ([z^{-1} , z]])$, where we have extended the commutative product $\bullet$ to $\Cal A$ by $k[z^{-1} , z]]$-linearity. We equipped $\Cal A$ by the minimal renormalization scheme:
\begin{equation}
{\Cal A} = {\Cal A}_- \oplus {\Cal A}_+,
\end{equation}
$$ \text{where:} \;\; {\Cal A}_+ : =   \mop {End} \wt {\Cal B} [[ z ]] \;\;\;\; \text{and} \;\;\;\; {\Cal A}_- : = \mop {End} z^{-1} \wt {\Cal B} [z^{-1}].$$  
We denote by $P$ the projection on ${\Cal A_-}$ parallel to ${\Cal A}_+$. 
\begin{theorem}
\begin{enumerate}
\item \hskip-0.1cm  Any character $\varphi$ has a unique Birkhoff decomposition in $G$: 
\begin{equation}
\varphi = \varphi _-^{\divideontimes-1} \divideontimes \varphi_+
\end{equation}
compatible with the renormalization scheme chosen.
\item The components $\varphi_+$ and $\varphi_- $ are given by simple recursive formulas: for all $(\bar\Gamma,\bar\gamma)$ of degree zero, $\varphi _-(\bar\Gamma,\bar\gamma) = \varphi _+(\bar\Gamma,\bar\gamma) = \varphi (\bar\Gamma,\bar\gamma) = Id_{\Cal B}$, and for all $(\bar\Gamma,\bar\gamma)$ of degree $n$ we put:
\begin{equation}\label{fin1}
\varphi_- (\bar\Gamma,\bar\gamma) =  - P \Big( \varphi (\bar\Gamma,\bar\gamma) + \sum_{\substack{\bar\delta \subsetneq\; \bar\gamma \\ \bar\gamma / \bar\delta \; \in \Cal T}} \varphi (\bar\Gamma/\bar\delta, \bar\gamma/\bar\delta)  \circ \varphi_-(\bar\Gamma, \bar\delta) \Big) 
\end{equation}
\begin{equation}\label{fin2}
\varphi_+ (\bar\Gamma,\bar\gamma) = (I- P) \Big( \varphi (\bar\Gamma,\bar\gamma) + \sum_{\substack{\bar\delta \subsetneq\; \bar\gamma \\ \bar\gamma / \bar\delta \; \in \Cal T}} \varphi (\bar\Gamma/\bar\delta, \bar\gamma/\bar\delta)  \circ \varphi_-(\bar\Gamma, \bar\delta) \Big) 
\end{equation}
\item $\varphi_+$ and $\varphi_-$ are two characters. We will call $\varphi_+$ the renormalized character and $\varphi_-$ the character of the counterterms.
\end{enumerate}
\end{theorem}
\begin{proof}
\begin{enumerate}
\item The existence of the Birkhoff decomposition of $\varphi$ is given by the formulas \eqref{fin1} and \eqref{fin2}, and we have $\varphi = \varphi _-^{\divideontimes-1} \divideontimes \varphi_+$. We now prove the uniqueness: We assume that $\varphi$ admits two decompositions i.e: 
 $$\varphi  =  \varphi_-^{-1}  \divideontimes \varphi_+ = \psi_-^{-1}  \divideontimes \psi_+.$$
Then we obtain the following equation: \;\;\; $\varphi_+ \divideontimes \psi_+^{-1} = \varphi_- \divideontimes \psi_-^{-1}.$\\
As for all $(\bar\Gamma,\bar\gamma) \in {\Cal D}_{\Cal T}$ we have: $\varphi_+ \ast \psi_+^{-1} (\bar\Gamma,\bar\gamma) \in {\Cal A}_+$ and $\varphi_-\divideontimes \psi_-^{-1} (\bar\Gamma,\bar\gamma) \in {\Cal A}_-$, then: $$\varphi_+ \divideontimes \psi_+^{-1} = \varphi_-\divideontimes \psi_-^{-1} = E,$$
$$\hskip-6.5cm \text {and consequently:} \;\;\;\;\;\; \;\;\;\;\;\;\;\;\;\;\;\;\;\;\;\;\;\;  \varphi_+ =  \psi_+ \;\;\; \text{and}  \;\;\; \varphi_- =  \psi_-.$$
\item The proof is lengthy but straightforward, similar to the proof of Theorem \ref{dt}.
\end{enumerate}
\end{proof}
From formulas \eqref{fin1} and \eqref{fin2} we get the algebraic frame explaining Smirnov's approach \cite[\S 8.2]{Sm}. Finally, this definition allows us to make sense of the renormalized Feynman rules.
\begin{definition}
The Feynman rules define an element $\wt I$ of $G$:
$$\wt I : \wt{\Cal D}_{\Cal T} \longrightarrow \Cal A \;\;\;\; (\Gamma, \gamma) \longmapsto \wt I (\Gamma, \gamma) :=  \wt I_{\Gamma, \gamma}^D,$$
$$\text{such that: }\;\;\;\;\;\;\;\wt I = \wt I _-^{\divideontimes-1} \divideontimes \wt I_+,$$
where $\wt I_-$ is the character of the counterterms. The renormalized character is $\wt I_+$ evaluated at $D =d$.
\end{definition}

\end{document}